\documentclass[a4paper, english,cleveref, autoref]{lipics-v2021}
\usepackage{booktabs}
\usepackage{xspace}
\usepackage{thmtools}
\usepackage{thm-restate}

\usepackage{algorithm}
\usepackage[noend]{algpseudocode}

\title{Range Counting Oracles for Geometric Problems}

\graphicspath{ {figures/} }

\author{Anne Driemel}{University of Bonn, Germany}{driemel@cs.uni-bonn.de}{}{}
\author{Morteza Monemizadeh}{Department of Mathematics and Computer Science, TU Eindhoven, the Netherlands}{m.monemizadeh@gmail.com}{}{}
\author{Eunjin Oh}{Department of Computer Science and Engineering, POSTECH, Korea}{eunjin.oh@postech.ac.kr}{}{}
\author{Frank Staals}{Department of Information and Computing Sciences, Utrecht University, The Netherlands}{f.staals@uu.nl}{}{}
\author{David P. Woodruff}{Carnegie Mellon University, Pittsburgh, PA, USA}{dwoodruf@andrew.cmu.edu}{}{}

\newcommand{\eps}{\ensuremath{\varepsilon}\xspace}
\newcommand{\R}{\ensuremath{\mathbb{R}}\xspace}
\DeclareMathOperator*{\polylog}{polylog}

\renewcommand{\Pr}[1]{\ensuremath{\mathbb{P}\left[#1\right]}}
\newcommand{\Ex}[1]{\ensuremath{\mathbb{E}[#1]}}

\newcommand{\Var}{\ensuremath{\mathsf{Var}}}

\newcommand{\opt}{\ensuremath{\mathsf{OPT}}}
\newcommand{\emd}{\ensuremath{\mathsf{EMD}}}
\newcommand{\mst}{\ensuremath{\mathsf{MST}}}
\newcommand{\sol}{\ensuremath{M}}

\newcommand{\llong}{\ensuremath{\text{long}}}

\newcommand{\parent}{\ensuremath{\mathsf{p}}}
\newcommand{\cell}{\ensuremath{c}}
\newcommand{\qt}{\ensuremath{\mathcal Q}}

\authorrunning{A. Driemel, M. Monemizadeh, E. Oh, F. Staals, D. Woodruff} 

\Copyright{Anne Driemel, Morteza Monemizadeh, Eunjin Oh, Frank Staals, David P. Woodruff} 



\category{} 

\acknowledgements{This research was initiated during the Workshop on Big Data and Computational Geometry, held at the Hausdorff Center for Mathematics (Bonn, Germany) in September 2024, and we gratefully acknowledge the workshop’s organizers and participants for their contributions to the fruitful discussions.
}

\nolinenumbers 

\hideLIPIcs  

\EventEditors{Oswin Aichholzer and Haitao Wang}
\EventNoEds{2}
\EventLongTitle{41st International Symposium on Computational Geometry
(SoCG 2025)}
\EventShortTitle{SoCG 2025}
\EventAcronym{SoCG}
\EventYear{2025}
\EventDate{June 23--27, 2025}
\EventLocation{Kanazawa, Japan}
\EventLogo{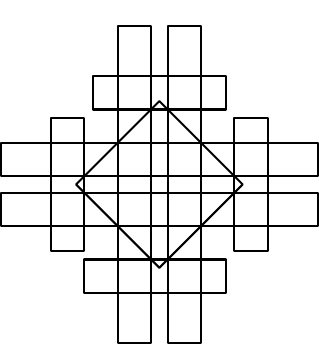}
\SeriesVolume{332}
\ArticleNo{XX}     

\begin{document}

\maketitle
\begin{abstract}
In this paper, we study estimators for geometric optimization problems in the sublinear geometric model. In this model, we have oracle access to a point set with size $n$ in a discrete space $[\Delta]^d$, where queries can be made to an oracle that responds to orthogonal range counting requests. The query complexity of an optimization problem is measured by the number of oracle queries required to compute an estimator for the problem.
We investigate two problems in this framework, the Euclidean Minimum Spanning Tree (MST) and Earth Mover Distance (EMD). For EMD, we show the existence of an estimator that approximates the cost of EMD with $O(\log \Delta)$-relative error and $O(\frac{n\Delta}{s^{1+1/d}})$-additive error using $O(s\polylog \Delta)$ range counting queries for any parameter $s$ with $1\leq s \leq n$. Moreover, we prove that this bound is tight.
For MST, we demonstrate that the weight of MST can be estimated within a factor of $(1 \pm \eps)$ using $\tilde{O}(\sqrt{n})$ range counting queries.
\end{abstract}

\section{Introduction}
\label{sec:Introduction}
In recent years, the size of data encountered in various applications has grown exponentially. While classical algorithms are designed to process the entire input to get exact or approximate solutions, the increasing demand for \emph{efficiency} in massive datasets necessitates approaches that can provide useful results without examining all of the input data. 
Motivated by this, 
a wide range of sublinear-time algorithms has been extensively studied over the past few decades from various subfields in theoretical computer science. 
For instance, there are sublinear-time algorithms for the longest increasing subsequence and edit distance problems on strings~\cite{kociumaka2020sublinear,mitzenmacher2021improved,goldenberg2019sublinear}, 
the triangle counting and vertex cover problems on graphs~\cite{behnezhad2023sublinear,eden2018approximating,onak2012near}, 
the Earth Mover's Distance problem on probability distributions~\cite{ba2011sublinear}, and the minimum spanning tree problem on general metric spaces~\cite{czumaj2004estimating}. Moreover, there are sublinear-time algorithms for 
several fundamental geometric problems~\cite{chazelle2003sublinear,czumaj05approx_weigh_euclid_minim_spann,czumaj2001property,har2021active,DBLP:conf/approx/Monemizadeh23}. 
For further information, refer to~\cite{czumaj2010sublinear,rubinfeld2011sublinear}. 

To obtain sublinear-time algorithms, we need an \emph{oracle}
to access the input data. 
There are various sublinear-time oracle models in the geometric setting although there are common oracles for strings, graphs and metric spaces.
Chazelle et al.~\cite{chazelle2003sublinear} presented several geometric algorithms assuming that the input is given without any preprocessing. In the case of point inputs, the only thing one can do is uniform sampling. In the case of polygons, one can also check if two edges are adjacent. Although they showed that several problems admit sublinear-time algorithms in this model, it seems not strong enough to solve a wider range of fundamental geometric problems such as the clustering problems, the Earth Mover's Distance problem, and the minimum spanning problem. 
Subsequently, many researchers developed sublinear-time algorithms for several different models. Examples are models in which the oracle can answer orthogonal range emptiness queries and cone nearest point queries~\cite{czumaj05approx_weigh_euclid_minim_spann}, orthogonal range counting queries~\cite{czumaj2001property,DBLP:conf/approx/Monemizadeh23}, or separation queries~\cite{har2021active}.

\medskip
In this paper, we study the \emph{range counting oracle model} for geometric optimization problems.
Here, we do not have direct access to the input point set, but instead we use an orthogonal range counting data structure for access. More specifically, given an orthogonal range as a query, it returns the number of input points contained in the query range. 
In fact, many database software systems, such as Oracle\footnote{\url{https://docs.oracle.com/en/database/other-databases/nosql-database/24.3/sqlreferencefornosql/operator1.html}} and Amazon SimpleDB\footnote{\url{https://docs.aws.amazon.com/AmazonSimpleDB/latest/DeveloperGuide/RangeQueriesSelect.html}}, provide built-in support for such range queries. These queries allow users to compute aggregates, such as the count, sum, or average of records that fall within a specified range of values for a set of attributes, which is crucial for a wide array of applications. For example, in data analytics, range queries are used to calculate the number of sales within a specific date range or determine the total revenue from products within a given price interval. In geographic information systems (GIS), range queries help aggregate spatial data points within certain coordinate bounds, such as counting the number of locations within a specific radius of a given point. Furthermore, in machine learning, range queries are employed during data preprocessing to summarize statistics over selected subsets, which supports tasks such as data filtering and dimensionality reduction.

The following are desirable properties for a reasonable oracle model: the queries to the oracle should be efficient to implement, it is preferable for it to be supported by lots of well-known databases, and it should allow sublinear-time algorithms for fundamental geometric problems.
First, this model clearly satisfies the first property; there are numerous works on data structures for range counting queries and their variants~\cite{afshani2009orthogonal,arya2000approximate,chan2011orthogonal,chan2018dynamic,chan2019orthogonal,de2008computational,sheng2011new}. In particular, one can preprocess a set of $n$ points in $\R^d$ in $O(n\log^{d-1}n)$ time to construct a data structure of size $O(n\log^{d-1}n)$ so that  the number of points inside a query range can be found in $O(\log^{d-1}n)$ time. 
Second, range counting queries are already supported by well-known public databases as we mentioned earlier. 
Therefore, without any special preprocessing, we can apply sublinear-time algorithms designed on the range counting oracle for such databases in sublinear time.
Third, we show that several fundamental problems can be solved in sublinear time on the range counting oracle model  in this paper.

\medskip 
We focus on the following three fundamental geometric problems: the Earth Mover's Distance problem, the minimum spanning tree problem, and the cell sampling problem.
The Earth Mover's Distance $\emd(R,B)$ between two point sets $R$ and $B$ in $\mathbb{R}^d$ and the cost $\mst(P)$ of a minimum spanning tree of a point set $P\in\mathbb{R}^d$ are defined as follows.
\[
\emd(R,B) = \min_{\pi: R\rightarrow B} \sum_{r\in R} \|r-\pi(r)\|\ \text{ and } \mst(P) = \min_{T} \sum_{uv\in E(T)} \|u-v\|,
\] 
where $\pi$ ranges over all one-to-one matchings between $R$ and $B$, and $T$ ranges over all spanning trees of $P$. 
If we have direct access to the input point set(s), then 
$\emd(R,B)$ can be computed exactly in near quadratic time in $\mathbb{R}^2$~\cite{agarwal2019efficient},
and approximately within $(1+\eps)$-relative error in near linear time~\cite{agarwal2022deterministic}. 
Similarly, $\mst(P)$ can be computed exactly in near quadratic time by computing the complete graph on $P$ and applying Kruskal's algorithm, and computed approximately in near linear time~\cite{har2011geometric}. 
Note that $\emd$ can be used for measuring the similarity between two point sets, and $\mst$ can be used for summarizing the distribution of a point set. 
Thus these problems have various applications in multiple areas of computer science including 
computer vision~\cite{bonneel2011displacement}, machine learning~\cite{flamary2018wasserstein} and document similarity~\cite{kusner2015word}. 

The cell sampling problem is defined as follows.
Let $P$ be a set of points in a discrete space $[\Delta]^d$.
Given a value $r$ with $1\leq r\leq \Delta$, let $\mathcal G$
be the grid imposed on $[\Delta]^d$ whose cells have side length $r$. 
We say a grid cell is \emph{non-empty} if it contains a  point of $P$. 
Our goal is to sample one non-empty cell almost uniformly at random from the set of all non-empty cells of $\mathcal G$.
If we have direct access to $P$, this problem is trivial as we can compute all non-empty cells explicitly. 
However, in our sublinear model, we cannot compute all non-empty cells if the number of non-empty cells exceeds our desired query complexity.  
We believe that cell sampling can be considered as a fundamental and basic primitive in sublinear models
as sublinear algorithms rely on efficient sampling to extract meaningful information from large data sets. 
In fact, this problem has been studied in 
the dynamic streaming model, and has been used as a primitive for several geometric problems~\cite{frahling2005sampling,DBLP:conf/soda/MonemizadehW10} 
and graph problems~\cite{DBLP:conf/esa/CormodeJMM17}. 

\begin{table}[]
    \centering
    \begin{tabular}{c|c|c|c|c}
    \hline
           &   & Additive Error  &  Multiplicative Error & Query Complexity    \\
     \hline\hline
  \multirow{2}{3em}{\emd} & UB &     $O(n\Delta/s^{1+1/d})$          & $O(\log\Delta)$  &   $\tilde O(s)$  \\
  & LB &           $\Omega(n\Delta/s^{1+1/d})$ & - & $\tilde O(s)$        \\
       \hline\hline
       \multirow{2}{6em}{\textsf{Cell Sampling}}& UB & - & $1\pm\eps$ & $\tilde O(\sqrt n)$      \\
       & LB & - & $O(1)$ & $\Omega(\sqrt n)$      \\
       \hline\hline 
       \multirow{2}{3em}{\mst} & UB  & - & $1\pm\eps$ & $\tilde O(\sqrt n)$\\
       & LB & - & $O(1)$ & $\Omega(n^{1/3})$ \\
       \hline            
    \end{tabular}                 
     \vspace{1em}
        \caption{Summary of our results. Here, $n$ denotes the number of points, and $\Delta$ denotes the size of the domain. For $\emd$, a parameter $s$ determines a trade-off between the additive error and the query complexity. 
}\label{tab:our-results}
\vspace{-1em}
\end{table}

\subparagraph{Our results.}
We present sublinear-time algorithms for \emd, \textsf{Cell Sampling}, and $\textsf{MST}$ in the orthogonal range counting oracle model.
Our results are summarized in Table~\ref{tab:our-results}. 

\begin{itemize}
    \item For \emd, we present an algorithm for estimating the Earth Mover's Distance between two point sets of size $n$ in $[\Delta]^d$ within a multiplicative error of $O(\log\Delta)$ in addition to an additive error of $O(n\Delta/s^{1+1/d})$. Our algorithm succeeds with a constant probability, and uses $\tilde O(s)$ range counting queries for any parameter $s$.
Notice that the optimal solution can be as large as $\Theta(n\Delta)$, and in this case, the additive error is relatively small. 
Furthermore, we show that this trade-off is tight for any parameter $s$.
    \item For \textsf{Cell Sampling}, we present an algorithm that samples a non-empty cell in a grid so that the probability that a fixed non-empty cell is selected is in $[(1-\eps)/m, (1+\eps)/m]$ for any constant $0<\eps<1$. Here $m$ denotes the number of non-empty cells. We also show that the query complexity here is tight. Moreover, if we sample $\tilde O(\sqrt{n}\log(n))$ non-empty cells, 
        we can estimate the number of non-empty cells within $(1\pm \eps)$-factor.
    \item For $\mst$, we present an algorithm for estimating the cost of a minimum spanning tree of a point set of size $n$ in $[\Delta]^2$ 
    within a multiplicative error of $1\pm\eps$. Our algorithm again succeeds  
    with constant probability, and  
uses $\tilde O(\sqrt n)$ range counting queries.
Also, we show that any randomized algorithm using $o(n^{1/3})$ queries has at least a constant multiplicative error.
\end{itemize} 

\subparagraph{Related work.}
Both $\textsf{EMD}$ and $\textsf{MST}$ have been extensively studied across various sublinear models over the past few decades. 
The work most closely related to ours is the one by Czumaj et al.~\cite{czumaj05approx_weigh_euclid_minim_spann}. 
They studied the Euclidean minimum spanning tree problem in a different model. In particular, their oracle supports \emph{cone nearest} queries, \emph{orthogonal range emptiness} queries, and \emph{point sampling} queries.\footnote{More precisely, the oracle supports cone nearest queries and orthogonal range queries only, but they are also given the access to the input point set, which allow them to sample a vertex uniformly at random.} 
Here, given a cone nearest query consisting of a cone $C$ with apex $p$, the oracle returns $(1+\delta)$-approximate nearest neighbor of $p$ in $C\cap P$, where $P$ is the set of input points. 
Using $\tilde O(\sqrt n)$ queries, their algorithm estimates the cost of the $\textsf{MST}$ of $n$ points within a relative error of $(1\pm \eps)$.
Moreover, they showed that if the oracle supports orthogonal emptiness queries only, then any deterministic algorithm for estimating the cost of $\textsf{MST}$
within $O(n^{1/4})$-relative error uses $\Omega(\sqrt n)$ orthogonal range emptiness queries. 
Sublinear time algorithms for computing a metric minimum spanning tree problem have also been studied~\cite{czumaj09estim_weigh_metric_minim_spann}. Here,
the oracle supports \emph{distance} queries: Given two points, it returns their distance in the underlying space.
The most popular sublinear-\emph{space} model is the streaming  model. In this model, the input arrives as data stream, and an algorithm is restricted to using a sublinear amount of memory relative to the input size. 
Both $\textsf{EMD}$ and $\textsf{MST}$ are studied extensively in this model~\cite{andoni2009efficient,chen2023streaming,frahling2005sampling}, showcasing the challenges and techniques for processing large datasets in the steaming model.

Although no sublinear algorithms for $\textsf{EMD}$ and $\textsf{MST}$ were known in the range counting oracle model prior to our work, other geometric optimization problems have been studied in this model.
Monemizadeh~\cite{DBLP:conf/approx/Monemizadeh23}
presented an algorithm for the facility location problem using $\tilde{O}(\sqrt n)$ queries, and Czumaj and Sohler~\cite{czumaj2001property} studied clustering problems and map labeling problems.
To the best of our knowledge, these are the only work done in the range counting oracle model. However, there are several closely related models such as the range-optimization model.
For instance, Arya et al.~\cite{arya2015approximate} studied the range MST problem, where the goal is to construct a data structure for computing the cost of MST within a query orthogonal range. In this setting, clustering problems also have been considered~\cite{abrahamsen2017range,oh2018approximate}.

\subparagraph{Preliminaries.}
We make use of quadtrees for both EMD and MST. 
Consider a discrete space $[\Delta]^2$.
For $i \in \{0, 1, \dots, \log \Delta\}$, we use $\mathcal{G}_i$
to denote a grid over the discrete space $[\Delta]^2$ consisting of cells with side lengths of $2^i$. A quadtree $\qt$ on $[\Delta]^2$ is then constructed as follows. The root of the quadtree corresponds to the unique cell
of $\mathcal G_{\log\Delta}$. For each node $v$ of the quadtree of depth $i$ has four equal sized children, each corresponding to a cell in $\mathcal G_{i-1}$ contained in the cell corresponding to $v$. Note that the depth of $\qt$ is $O(\log\Delta)$.
If it is clear from the context, we use a node of the quadtree and its corresponding cell interchangeably. 
For MST, we use the quadtree as it is, but for EMD, we need a randomly shifted quadtree. Here, we select a vector from $[\Delta]^2$ uniformly at random, and we shift all the grids $\mathcal G_i$'s by the vector. 
Then any rectangle shifted by the vector can be represented as the union of four orthogonal rectangles in $[\Delta]^2$. Therefore, without loss of generality, we assume that the chosen vector is the origin. 

Now we introduce two basic operations we use for both EMD and MST. 
Let $P$ be a point set in $[\Delta]^2$. 
We can find the lexicographically $k$-th smallest point in a given query range using $O(\log\Delta)$ range counting queries: apply binary search for the $x$-coordinates and then apply binary search for the $y$-coordinates. 
Here, for any two points $p$ and $q$, we say $p$ is lexicographically smaller than $q$ if $p_x < q_x$, or $p_x=q_x$ and $p_y < q_y$, where $p_x$ and $p_y$ denote the $x$-coordinate and the $y$-coordinate of $p$, respectively.
Similarly, we can find the $k$-th smallest point according to the \emph{$z$-order}  in a given query range using $O(\log\Delta)$ range counting queries.
Here the $z$-order of a point set $P$ is defined recursively on the quadtree $T$ as follows. 
For each cell $c$, we do the following. Let $c_1,c_2,c_3$ and $c_4$ be the children of $c$ sorted in the lexicographical order of their centers. 
Then the $z$-ordering of $P$ is the concatenation of the $z$-ordering of $P\cap c_i$ for $i=1,2,3,4$.


Let $P$ be a point set in $[\Delta]^d$. 
%
%
Furthermore, we will often want to sample a point from $P$ uniformly at random. We can do so using $O(\log\Delta)$ range counting queries in a procedure called \emph{telescoping sampling}~\cite{DBLP:conf/approx/Monemizadeh23}.

\begin{lemma}[Telescoping Sampling~\cite{DBLP:conf/approx/Monemizadeh23}]\label{lem:telescoping}
    We can select a point of $P$ uniformly at random
    using $O(\log\Delta)$ range counting queries.
\end{lemma}

Throughout this paper, we use $\tilde O(f(s))$ to denote $O(f(s)\polylog \Delta)$. 







\section{Sublinear algorithms for the Earth Mover's Distance problem}

In this section, we present several sublinear-time algorithms  that approximate the cost of the Earth Mover's Distance between two point sets in $[\Delta]^d$ for an integer $d\geq 1$.
The approximation bounds we obtain are tight up to a factor of $O(\log\Delta)$ for any integer $d\geq 1$.

In a one-dimensional space, we know the configuration of an optimal matching: the $k$th leftmost red point is matched with the $k$th leftmost blue point for all $1\leq k\leq n$. 
Thus it suffices to estimate the cost of this matching.
On the other hand, we do not know the exact configuration of an optimal matching in a higher dimensional space.
Thus instead of considering an optimal matching, we consider a matching of cost  $O(\opt\log\Delta)$, 
where $\opt$ is the Earth Mover's Distance between two point sets. 
This is where the multiplicative error of $O(\log\Delta)$ comes from. 
Apart from this,
we use the same approach for all dimensions $d\geq 1$.

\subsection{A sublinear algorithm for points in 1D}
In this section, we present a sublinear algorithm using $\tilde O(s)$ range counting queries in a one-dimensional space. 
More precisely, we prove the following theorem. 
The optimal solution can be as large as $\Theta(n\Delta)$, and in this case, the additive error is relatively small. 

\begin{theorem}
    \label{thm:emd:additive}
    Given two sets $R$ and $B$ of size $n$ in a discrete space $[\Delta]$ and a parameter $s>1$, 
    we can estimate the cost of the Earth Mover's Distance between $R$ and $B$ within $O(1)$ multiplicative error in addition to  
    $O(\frac{\Delta n}{s^2})$ additive error with probability at least $2/3$ using $\tilde O(s)$ range counting queries.
\end{theorem}

If no two points coincide, we have $\opt\geq n$. Therefore, we have the following corollary. 
\begin{corollary}
    Given two sets $R$ and $B$ of size $n$ in a discrete space $[\Delta]$ such that no two points of $R\cup B$ coincide, and a parameter $s>1$, 
    we can estimate the cost of the Earth Mover's Distance between $R$ and $B$ within 
    $O(\frac{\Delta}{s^2})$ multiplicative error with probability at least $2/3$ using $\tilde O(s)$ range counting queries.
\end{corollary}

Let $R=\{r_1,r_2,\ldots,r_n\}$ and $B=\{b_1,b_2,\ldots,b_n\}$ be two sets in $[\Delta]$ sorted along the axis.
Let $M$ be an optimal matching between $R$ and $B$.
We separate two cases: \emph{long} edges and \emph{short} edges.
An edge of $M$ is \emph{long} if its length is at least $\Delta/s$, and it is \emph{short}, otherwise.
Let $\opt$ be the Earth Mover's Distance between $R$ and $B$.

We employ the following \emph{win-win} strategy:
For a long edge, a slight perturbation of its endpoints introduces only a constant relative error. Thus, we move its endpoints to the center of a grid, enabling us to estimate the total length efficiently.
For short edges, we can disregard cases where the number of such edges is small, as this results in only a small additive error. Consequently, we may assume that the number of short edges is large, allowing us to estimate their total length using sampling. To further reduce the additive error, we partition the set of short edges into $\log (\Delta/s)$ subsets based on edge length and then consider each subset separately.


\subparagraph{Long edges.}
For long edges, 
we subdivide $[\Delta]$ into segments each of length $\Delta/(2s)$.
For each long edge $(r_i,b_i)$, 
imagine that we move $r_i$ and $b_i$ to the centers of the segments containing them. Since $(r_i,b_i)$ is long, this introduces  an additive error of $O(\Delta/s)$, resulting in a constant relative error.
Using this fact, we estimate the total length of all long edges within a constant relative error using $\tilde O(s)$ range counting queries. 
Let $L_\llong$ be the resulting estimator. 

\begin{lemma}
    We can estimate the total length of the long edges in $M$ within $O(1)$ multiplicative error using $\tilde O(s)$ range counting queries. 
\end{lemma}
\begin{proof}
    We subdivide $[\Delta]$ into $s$ segments each of length $\Delta/(2s)$. Here, the number of segments is $O(s)$. 
    Let $\mathcal S$ be the set of the resulting segments. 
    Then a long edge contains at least one segment. Imagine that we move every point $p$ incident to a long edge to the center of the segment containing $p$.
    For each segment $S$, we compute the number $n_r(S)$ of red points contained in $S$ and the number $n_b(S)$ of blue points contained in $S$, respectively, using $O(s)$ queries. 
    Then we construct another point set $P'$ by adding $n_r(S)$ red points to the center of $S$ and adding $n_b(R)$ blue points to the center of $S$. We can do this without using queries. For a point $p\in R\cup B$, let $p'$ be the its corresponding point in $P'$. 
    Then consider the matching $M'=\{(r',b') \mid (r,b)\in M\}$. 
    We can compute $M'$ explicitly without using any queries. Let $M'_\llong$ be the set of edges of $M'$ containing at least one segment of $\mathcal S$. 
    A long edge of $M$ has its corresponding edge in $M'_\llong$. However, some short edge of $M$ might have its corresponding edge in $M'_\llong$.
    Here, observe that the short edges of $M$ having their corresponding edges in $M'_\llong$ have length at least $\Delta/(2s)$, thus the cost induced by those edges are also within a constant factor of $\opt$. 
    They will be also counted for short edges, but this also increases the total cost by a constant factor.
    Therefore, the lemma holds. 
\end{proof}

\subparagraph{Short edges.}
For short edges, we use random sampling. 
More specifically, we subdivide the set of short edges into $t=\log (\Delta/s)$ subsets with respect to their lengths: $E_1,E_2,\ldots,E_t$,
where $E_i$ is the set of all edges of $M$ of lengths 
in $[2^{i-1}, 2^{i})$. 
Then we estimate the \emph{number} of edges of $E_i$ for every index $i$.
To do this, we choose a random sample of $R$ of size $x=O(s\log^2\Delta)$. For each chosen point, we can find its mate. Let $S_i$ denote the set of all edges between the sampled points and their mates. 
If $|E_i| \geq n/(s\log \Delta)$ then the number of edges in $E_i \cap S_i$ would be a good estimator for $|E_i|$.
Otherwise, 
we can ignore them as $E_i$ induces an additive error of at most $\Delta n/(s^2\log\Delta)$. This is because every short edges has length at most $\Delta/s$. 
Let $\ell_i = |E_i\cap S_i| \cdot (n/x)$.
Then we have the following lemma. 

\begin{lemma}\label{lem:success-prob-1d}
    If $|E_i|\geq n/(s\log\Delta)$,  we have $\ell_i=\Theta(|E_i|)$ with a constant probability. 
    Otherwise, we have $\ell_i\leq O(|E_i|)$ with a constant probability.
\end{lemma}
\begin{proof}
Let $Y_j$ be the random variable where $Y_j=1$ if the $j$th  edge in $S_i$ is contained in $E_i$, and $Y_j=0$ otherwise. 
Let $Y=Y_1+\ldots+Y_x$, where $x=s\log^2 \Delta$ is the number of sampled points.  
By definition, we have  $\Ex{Y_j}= \frac{|E_i|}{n}$, and $\Ex{Y}=x\cdot \frac{|E_i|}{n}$. 
Now we analyze the failure probability. If $|E_i| \geq n/(s\log\Delta)$, we use Chernoff bounds as stated in Lemma~\ref{lem:cher}. 
Then the probability that $|\Ex{Y}-Y| \geq \eps \Ex{Y}$
is at most $2\exp(-\eps^2 \cdot p\cdot (s\log^2\Delta)/3)$, 
where $p$ is the probability that $Y_j=1$. 
In our case, $p=|E_i|/n \geq 1/(s\log\Delta)$. 
Thus the failure probability is at most $1/\Delta$ in this case. 
If $|E_i|< n/(s\log\Delta)$, we use Markov's inequality as stated in Lemma~\ref{lem:markov}.
Then the probability that $Y_i \geq c\dot |E_i|$ is less than a small constant for a large constant $c$. Therefore, the second part of the lemma also holds.
\end{proof}

To amplify the success probability of Lemma~\ref{lem:success-prob-1d}, we repeat the procedure for computing $\ell_i$ $O(\log\Delta)$ times. Then we take the minimum, and 
let $L=\sum_{i=1}^t 2^i\cdot \ell_i + L_\llong$.
\begin{lemma}
    With a constant probability, $(1/4)\cdot \opt-n\Delta/s^2 \leq L \leq 4\cdot \opt+n\Delta/s^2$.
\end{lemma}
\begin{proof}
    Since we can estimate the total length of long edges with an additive error of $2\cdot\opt$, it suffices to focus on short edges. 
    For each edge of $E_i$, we round up its length to $2^i$.
    This induces the additive error of $2\cdot\opt$. 
    Now we use Lemma~\ref{lem:success-prob-1d}. Consider the event that we have the desired bounds for all indices $i$ in Lemma~\ref{lem:success-prob-1d}. Due to the amplification of success probability we made before,
    this total success probability is a constant. 
    If $|E_i|\geq (n/(s\log\Delta)$, the estimator $\ell_i$ 
    has an error of $O(1)\cdot \opt$ in total for all such indices.
    If $|E_i| < n/(s\log\Delta)$, we do not have a lower bound guarantee for $\ell_i$. 
    However, since the number of such edges is $n/s$ in total,
    and the length of each such edge is $\Delta/s$,
    this induces the additive error of $n\Delta/s^2$.
    Therefore, the total additive error is $n\Delta/s^2$. 
\end{proof}

\subsection{A sublinear algorithm for higher dimensions}
In this section, we present a sublinear algorithm that approximates the cost of the Earth Mover's Distance between two point sets in $[\Delta]^d$ for an integer $d\geq 2$. 
More specifically, we prove the following theorem. 
Again, notice that the Earth Mover's Distance between two point sets can be as large as $\Theta(\Delta n)$ in the worst case. 

\begin{restatable}{theorem}{emdalgo}
     \label{thm:emd:additive-d}
    Given two sets $R$ and $B$ of size $n$ in a discrete space $[\Delta]^d$ and a parameter $s>1$, 
    we can estimate the cost of the Earth Mover's Distance between $R$ and $B$ within $O(\log\Delta)$-relative error in addition to 
    $O(\frac{\Delta n}{s^{1+1/d}})$ additive error with probability at least $2/3$ using $\tilde O(s)$ range counting queries.
\end{restatable}

If no two points coincide, we have $\opt\geq n$. Therefore, we have the following corollary. 

\begin{corollary}
    \label{cor:emd:additive}
    Given two sets $R$ and $B$ of size $n$ in a discrete space $[\Delta]^2$ and a parameter $s>1$ such that no two points in $R\cup B$ coincide, 
    we can estimate the cost of the Earth Mover's Distance between $R$ and $B$ within $O(\max\{\log\Delta,\frac{\Delta}{s^{1+1/d}}\})$-relative error with probability at least $2/3$ using $\tilde O(s)$ range counting queries.
\end{corollary}

Let $\opt$ be the Earth Mover's Distance between $R$ and $B$. 
This algorithm is essentially the same as the algorithm for the one-dimensional case. The only difference is that 
we do not know the configuration of an optimal matching in a higher dimensional space.
Instead, we can use a matching obtained from the quadtree on $[\Delta]^d$, which has cost of $O(\log\Delta)\cdot \opt$.
Then as before, we separate two cases. 
If an edge has length at least $\Delta/s^{1/d}$, it is \emph{short}. Otherwise, it is \emph{long}.
Then we can compute all long edges exactly, and we can estimate short edges using sampling.

To make the description easier, we considers the case that $d=2$ only. It is not difficult to see that this algorithm indeed works for any higher dimensional space.

\subparagraph{Greedy matching $\sol$ with $O(\log \Delta)$ multiplicative error.}
We first define a greedy matching between $R$ and $B$ of cost $O(\log\Delta)\cdot\opt$.
For this, we define a new metric space $([\Delta]^2,\ell)$ based on a randomly shifted quadtree $T$ as follows.
First, we define the length to each edge of $T$. 
For each edge of the quadtree $T$ between level $i$ and level $i+1$, its length is defined as $2^{i+1}$. Recall that a leaf has level zero.  
Then, for any two points $p$ and $q$ in $[\Delta]^2$, 
we define $\ell(p,q)$ as the distance in the quadtree $T$ between the leaves corresponding to $p$ and $q$.
By construction, $\|p-q\| \leq \ell(p,q)$, where $\|p-q\|$ denotes the Euclidean distance between $p$ and $q$. Moreover,
by~\cite{DBLP:conf/stoc/Charikar02,DBLP:conf/stoc/Indyk04}, we have the following lemma.
This implies that the expected value of the Earth Mover's Distance between $R$ and $B$ under $([\Delta]^2, \ell)$ 
yields a multiplicative error of $O(\log\Delta)$. Also, notice that 
it is always at least $\opt$. This fact will be used later to analyze the success probability of the entire algorithm. 

\begin{lemma}[\cite{DBLP:conf/stoc/Charikar02,DBLP:conf/stoc/Indyk04}]\label{lem:approx-sol}
    For any two points $p$ and $q$ in $[\Delta]^2$, the expected value of $\ell(p,q)$
    is at most $O(\log \Delta) \cdot \|p-q\|$, where $\|p-q\|$ denotes the Euclidean distance between $p$ and $q$.
\end{lemma}


Thus we may focus on an optimal matching between $R$ and $B$ for the Earth Mover's Distance under $([\Delta]^2, \ell)$.
In fact, a greedy solution computed in a bottom-up fashion yields an optimal solution in this case as shown in~\cite{DBLP:conf/stoc/Charikar02,DBLP:conf/stoc/Indyk04}.
More specifically, we consider the cells $c$ in the quadtree in a bottom-up fashion.
It is partitioned into four cells $c_1,c_2,c_3$ and $c_4$ in the previous level.
As an invariant, assume that
the number of edges matched inside $c_i$ is $\max \{ n_r(c_i), n_b(c_i)\}$, where
$n_r(c_i)$ and $n_b(c_i)$ be the numbers of red and blue points in $c_i$, respectively.
Let $R_c$ and $B_c$ be the sequences of red points and blue points not matched inside $c_i$ for $i=1,2,3,4$, respectively, sorted in the $z$-order. 
Then we match the points in $R_c$ and $B_c$ as many as possible according to  their orderings in the sequences.
For a point $p\in B\cup R$, we call the point in $B\cup R$ matched with $p$ the \emph{mate} of $p$.
Let $M$ be the matching obtained in this way. 
It is known that $M$ is an optimal matching under $([\Delta]^2,\ell)$~\cite{DBLP:conf/stoc/Charikar02,DBLP:conf/stoc/Indyk04}.
Also, notice that 
$\ell(\sol) = \sum_{i=0}^{\log\Delta} \sum_{c \in \mathcal{G}_i} |n_r(c)  - n_b(c)|\cdot 2^{i+1}$.
In the following, we let $\ell(X)$ denote
the total length of all edges of $X$ for an edge set $X\subseteq B\times R$.

\begin{lemma}\label{lem:mate}
    We can find the mate of a given point using $O(\log \Delta)$ range counting queries.
\end{lemma}
\begin{proof}
    Let $b$ be a given point. Without loss of generality, assume that it is a blue point. We traverse the quadtree $T$ from the leaf corresponding to $b$ until we find its mate, say $r$. 
    Assume that we are at a node corresponding to a cell $c$.
    As an invariant, assume that the mate of $p$ is not matched in a descendant of $c$ (excluding $c$). In other words, no descendant of $c$
    contains both $b$ and $r$.
    Then our goal in this level of $T$ is to determine if $r$ is contained in $c$.
    If so, we want to return $r$.
    
    For this, let $c_1,c_2,c_3$ and $c_4$ denote the children of $c$
    sorted along the lexicographical order of their centers. 
    Notice that for each cell $c'$ of the quadtree, we can compute the number of unmatched blue points inside $c'$, which is
    $\max\{n_r(c_i)-n_b(c_i),0\}$. Symmetrically, we can compute the number of unmatched red points inside $c'$. 
    Using this, we can compute the number $k$ of unmatched blue points in $c_1,\ldots, c_4$ whose $z$-order is less than the $z$-order of $b$ using $O(\log\Delta)$ range counting queries. 
    If this is at least the number of unmatched red points in $c_1,\ldots,c_4$, then $b$ is not matched inside $c$. Thus we need to consider the ancestors. 
    Otherwise, $b$ is matched with the $(k+1)$th unmatched red points (along the $z$-order) in
    $c_1,\ldots, c_4$. Such a point can be found using $O(\log\Delta)$ range counting queries. 
\end{proof}

\subparagraph{A sublinear algorithm for estimating the cost of $\sol$.}
Now it suffices to estimate the cost of $\sol$. 
As before, we consider two cases. 
An edge of $\sol$ is \emph{short} if its length is less than $\Delta/\sqrt s$; otherwise it is a \emph{long} edge.
By the construction of $M$,
we can compute all long edges exactly using $O(s)$ range counting queries. This is because long edges are defined inside the grids of level at least $\log(\Delta/\sqrt s)$.
The number of all grid cells in those grids is $O(s)$. 
Let $L_\llong$ be the total length of long edges. 
Thus we can focus on estimating the total length of short edges of $M$.
Here, we again subdivide the set of all short edges into $t=\log (\Delta/\sqrt s)$ subsets: $E_1,E_2,\ldots,E_t$,
where $E_i$ is the set of all edges of $M$ of lengths 
in $[2^{i-1}, 2^{i})$. 
Then we estimate the number of edges of $E_i$ for every index $i$.
To do this, we choose a random sample of $R$ of size $x=O(s\log^2\Delta)$. For each chosen point, we can find its mate using Lemma~\ref{lem:mate}. Let $S_i$ denote the set of all edges between the sampled points and their mates. 
Then the number of edges in $E_i \cap S_i$ would be a good estimator for $|E_i|$ if $|E_i| \geq n/(s\log \Delta)$.
Otherwise, 
we can ignore it as it induces the additive error of at most $\Delta n/(s^{1.5}\log\Delta)$. This is because every short edges has length at most $\Delta/s^{0.5}$. 
Let $\ell_i = |E_i\cap S_i| \cdot (n/x)$.
Then we have the following lemma. 

\begin{lemma}\label{lem:success-prob-2d}
    If $|E_i|\geq n/(s\log\Delta)$,  we have $\ell_i=\Theta(|E_i|)$ with a constant probability. 
    Otherwise, we have $\ell_i\leq O(|E_i|)$ with a constant probability.
\end{lemma}
\begin{proof}
Let $Y_j$ be the random variable where $Y_j=1$ if the $j$th  edge in $S_i$ is contained in $E_i$, and $Y_j=0$ otherwise. 
Let $Y=Y_1+\ldots+Y_x$, where $x=s\log^2 \Delta$ is the number of sampled points.  
By definition, we have  $\Ex{Y}= x\cdot \frac{|E_i|}{n}$.
Now we analyze the failure probability. If $|E_i| \geq n/(s\log\Delta)$, we use Chernoff bounds as stated in Lemma~\ref{lem:cher}. 
Then the probability that $|\Ex{Y}-Y| \geq \eps \Ex{Y}$
is at most $2\exp(-\eps^2 \cdot p\cdot (s\log^2\Delta)/3)$, 
where $p$ is the probability that $Y_j=1$. 
In our case, $p=|E_i|/n \geq 1/(s\log\Delta)$. 
Thus the failure probability is at most $1/\Delta$ in this case. 
If $|E_i|< n/(s\log\Delta)$, we use Markov's inequality as stated in Lemma~\ref{lem:markov}.
Then the probability that $Y \geq c\dot |E_i|$ is less than a small constant for a large constant $c$. 
Therefore, the lemma holds.
\end{proof}

To amplify the success probability of Lemma~\ref{lem:success-prob-1d}, we repeat the procedure for computing $\ell_i$ $O(\log\Delta)$ times. Then we choose the minimum, and 
let $L=\sum_{i=1}^t 2^i\cdot \ell_i + L_\llong$.
\begin{lemma}
    With a constant probability,
    the estimator $\ell$ approximates $\opt$ within the relative error of $O(\log\Delta)$ and the additive error of $O(n\Delta/s^{1.5})$.
\end{lemma}
\begin{proof}
    Since we can compute the total length of long edges exactly, it suffices to focus on short edges. 
    For each edge of $E_i$, we round up its length to $2^i$.
    Here, we have the additive error of $2\cdot\ell(M)$. 
    Now we use Lemma~\ref{lem:success-prob-1d}. Consider the event that we have the desired bounds for all indices $i$ in Lemma~\ref{lem:success-prob-2d}. Due to the amplification of success probability we made before,
    this total success probability is a constant. 
    If $|E_i|\geq n/(s\log\Delta)$, the estimator $\ell_i$ 
    has an additive error of $O(1)\cdot \opt$ in total for all such indices by Lemma~\ref{lem:success-prob-2d}. 
    If $|E_i| < n/(s\log\Delta)$, we have $\ell_i\leq O(|E_i|)$.
    Since the number of such edges is $n/s$ in total,
    and the length of each such edge is $\Delta/s^{0.5}$,
    this induces the additive error of $n\Delta/s^{1.5}$. 
    Therefore, we can estimate the cost of $M$ 
the relative error of $O(1)$ and the additive error of $O(n\Delta/s^{1.5})$. 

    Thus assuming that $M$ is within a multiplicative error of $O(\log\Delta)$ from $\opt$, our estimator estimates $\opt$
    within the relative error of $O(\log\Delta)$ and the additive error of $O(n\Delta/s^{1.5})$. 
    Then we repeat the entire algorithm $O(\log\Delta)$ times and take the smallest estimator $L$. 
    Since the cost of $M$ is always at least $\opt$, we can use 
    Markov's inequality in Lemma~\ref{lem:markov} to analyze the success probability.
    Therefore, the lemma holds with a constant probability.
\end{proof}

\section{Lower bounds for the Earth Mover's Distance}
\label{sec:EMD} 


In this section, we show that for any parameter $s$, any randomized algorithm that approximates the Earth Mover's Distance between two sets within an additive error of $O(n\Delta/s^{1+1/d})$ requires at least $s$ orthogonal range counting queries.  
More specifically, we prove: 

\begin{lemma}\label{lem:lowerbound}
  For any parameter $s$ and any integers $d\geq 1$ and $\Delta\geq 1$, any randomized algorithm for approximating the Earth Mover's Distance between two point sets of size $n$ in a discrete space $[\Delta]^d$ within
  additive error $O(n\Delta/s^{1+1/d})$ uses $\Omega(s)$  range counting queries.
\end{lemma}

If points from $B\cup R$ may coincide, achieving any multiplicative approximation factor is impossible. Even when no two points of $B\cup R$ coincide, Lemma~\ref{lem:lowerbound} shows that no randomized algorithm with a multiplicative error of $O(\Delta/s^{1+1/d})$ uses $o(s)$ range counting queries.


\subsection{A lower bound for points in 1D}
In this section, we prove Lemma~\ref{lem:lowerbound} in the one-dimensional space. This will be extended for a higher dimensional space later. 
To analyze the lower bound on the additive error of a Monte Carlo algorithm, we use the following version of Yao's Minmax Theorem~\cite{arora2009computational}: 
    Assume that there exists a randomized algorithm using $s$ range counting queries with success probability at least $2/3$. Then for any distribution $\mu$ of instances, there exists a deterministic algorithm using $s$ range counting queries that returns desired answers on instances chosen from $\mu$ with probability at least $2/3$.
Here, the success probability of a randomized algorithm
is the probability that the output is within our desired bound. 

Thus it suffices to construct a distribution $\mu$ of instances
such that no deterministic algorithm using $s$ range counting queries has success probability at least $2/3$. 
For this, we define two types of gadgets on a segment $S$ of length $\Delta/(8s)$: the \emph{near} gadget and \emph{far} gadget. 
See Figure~\ref{fig:lowerbound}.
For the far gadget, we put 
$4n/s$ red points near the left endpoint of $S$ and $4n/s$ blue points near the right endpoint of $S$.
For the near gadget, we put 
$2n/s$ red points and $2n/s$ blue points alternately starting from the left endpoint of $S$ so that the distance between any two consecutive points is one, and the starting point lies in the left endpoint of $S$.
Additionally, we put 
$2n/s$ red points and $2n/s$ blue points alternately so that the distance between any two consecutive points is one, and the last point lies in the right endpoint of $S$. 
Notice that the cost induced by the near gadget is $\Theta(n/s)$, but the cost induced by the far gadget is $\Theta(n\Delta/s^2)$.
Our strategy is to place $O(s)$ copies of the near gadget inside $[\Delta]$ and to hide one far gadget inside $[\Delta]$ with probability $1/2$. 
In this way, we can construct two types of instances, one with cost $n$ and one with cost $\Theta(n\Delta/s^2)$.
\begin{figure}
    \centering
    \includegraphics[width=\linewidth]{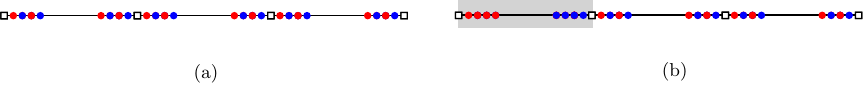}
    \caption{(a) All segments use the near gadget. The cost of this instance is $n$. (b) The gray segment uses the far gadget. The cost of this instance is $\Theta(n\Delta/s^2)$.}
    \label{fig:lowerbound}
\end{figure}

More specifically, we define a distribution $\mu$ of instances as follows.
We partition $[\Delta]$ into $8s$ segments each of length $\Delta/(8s)$. 
Let $\mathcal S=\{S_1,S_2,\ldots,S_{8s}\}$ be the set of resulting segments along the axis. See Figure~\ref{fig:lowerbound}. 
With probability $1/2$, we let $t=0$.
Then with probability $1/2$, we 
choose one index $t$ from $1,2,\ldots,s$ uniformly at random. 
That is, the probability that a fixed index $i$ is chosen is $1/(16s)$ for $i=1,2,\ldots,8s$. 
Then for each index $i\neq t$, we 
place the near gadget on $S_i$.
For index $i=t$, we place the far gadget on $S_i$. 
Notice that we do not place the far gadget anywhere if $t=0$. 
Thus it suffices to show that,
for any deterministic algorithm using $s$ range counting queries,
the probability that it
estimates the cost of $\opt$ within an additive error of $O(n\Delta/s^{2})$ 
on the instances chosen from $\mu$ is less than $2/3$. 

\begin{restatable}{lemma}{emdlb}\label{lem:distribution-1D}
    For any deterministic algorithm using $s$ range counting queries, 
    the probability that it estimates the cost of an instance chosen from $\mu$ within an additive error of $O(n\Delta/s^2)$ is less than $2/3$.   
\end{restatable}
\begin{proof}
    Assume to the contrary that such a deterministic algorithm $\mathcal A$ exists. 
    For instances in which we use the far gadget, its cost is $O(n\Delta/s^2)$. On the other hand, the instance for which we do not use the far gadget has cost exactly $n$. 
    Thus $\mathcal A$ must determine whether or not a given instance has the far gadget with probability at least $2/3$. 
    
    We can consider $\mathcal A$ as a decision tree
    of depth $s$. Here, each node corresponds to an orthogonal range, and each leaf corresponds to the output. 
    Given a query interval $Q$, we say that it \emph{hits} a segment $S$ of $\mathcal S$ if an endpoint of $Q$ is contained in $S$. 
    Note that the algorithm cannot get any information for the segments not hit by queries. 
    In particular, for any segment $S$ and any interval range $Q$, the number of red (and blue) points in $S\cap Q$ remains the same unless $S$ contains an endpoint of $Q$. 
    On the other hand, $\mathcal A$ uses only $s$ range counting queries, and thus the probability that any of the $s$ range queries hits the far gadget is less than $2/3$.

    More formally, we consider the event that an instance using the far gadget is chosen from $\mu$. This even happens with probability at least $1/2$ by construction. Thus to get a success probability at least $2/3$, we need to detect the far gadget. 
    Recall that we can detect the far gadget only when it is hit by $Q$.  On the other hand, each query hits at most two segments. Notice that the probability that a fixed query range hits the far gadget is $2/(8s)=1/(4s)$ as the total number of segments in $\mathcal S$ is $8s$. 
    Therefore, the probability that 
    at least one of the $s$ queries hits the far gadget is $1-(1-(4s)^{-1})^{s}) < 1/3$. In other words, the failure probability is at least $2/3$ assuming that we are given an  instance using the far gadget. 
    Since the probability that an instance using the far gadget is chosen is $1/2$, the total failure probability is larger than $1/3$, and thus the lemma holds.
\end{proof}

\subsection{Lower bounds for dimensions 2 and higher}
We now extend the construction of the 1D case to the 2D case. 
We again use Yao's Minmax Theorem.
Thus it suffices to construct a distribution $\mu$ of instances
such that no deterministic algorithm using $s$ range counting queries has success probability at least $2/3$. 

For this, we define two types of gadgets on a square $S$ of side length $\Delta/(8\sqrt{s})$: the \emph{near} gadget and \emph{far} gadget. 
See Figure~\ref{fig:lowerbound-2D}.
For the far gadget, we put two copies of the far gadget constructed from the 1D case, one on the upper side of $S$ and one on the lower side of $S$, so that a red point comes first in the upper side, and a blue point comes first in the lower side. 
For the near gadget, we do the same: Put two copies of the near gadget constructed from the 1D case, one on the upper side and one on the lower side, so that a red point comes first in the upper side, and a blue point comes first in the lower side. 
In this way, every gadget has $8n/s$ red points and $8n/s$ blue points.

\begin{figure}
    \centering
    \includegraphics{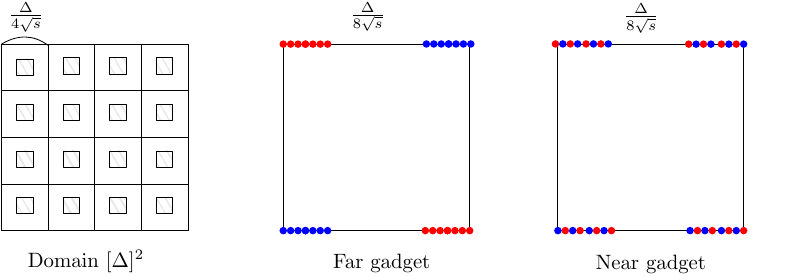}
    \caption{We partition $[\Delta]^2$ into $16s$ squares. In each square, we place either the far gadget or the near gadget.  
    The cost of the far gadget is at least $n\Delta/s^2$,
    and the cost of the near gadget is $\Theta(n/s)$.}
    \label{fig:lowerbound-2D}
\end{figure}

Now define a distribution $\mu$ of instances as follows.
We construct the grid on $[\Delta]^2$ consisting of $16s$ cells  each of side length $\Delta/(4\sqrt s)$. See Figure~\ref{fig:lowerbound-2D}. 
We choose a cell $W$ as follows. 
With probability $1/2$, we let $W=\emptyset$. 
Then with probability $1/2$, we 
choose one cell from the $16s$ cells uniformly at random, 
and let it $W$. 
Then for each cell in the grid, except for $W$, we 
place the near gadget in the middle of the cell.
Then we place the far gadget in the middle of $W$. 
Notice that we do not place the far gadget anywhere if $W=\emptyset$. 
Thus it suffices to show that,
for any deterministic algorithm using $s$ range counting queries,
the probability that it
estimates the cost of $\opt$ within an additive error of $O(n\Delta/s^{1.5})$ 
on an instance chosen from $\mu$ is less than $2/3$. 
We can prove this similarly to Lemma~\ref{lem:distribution-1D}.

\begin{restatable}{lemma}{lbtwo}\label{lem:distribution-2D}
    For any deterministic algorithm using $s$ range counting queries, 
    the probability that it estimates the cost of an instance chosen from $\mu$ within an additive error of $O(n\Delta/s^{1.5})$ is less than $2/3$. 
\end{restatable}
\begin{proof}
    Assume to the contrary that such a deterministic algorithm $\mathcal A$ exists. 
    For instances we use the far gadget, its cost is $\Theta(n\Delta/s^{1.5})$. On the other hand, the instance we do not use the far gadget has cost is $\Theta(n)$. 
    Thus $\mathcal A$ must determine whether or not a given instance has the far gadget with success probability at least $2/3$. 
    
    We can consider $\mathcal A$ as a decision tree
    of depth $s$. Here, each node corresponds to an orthogonal range, and each leaf corresponds to the output. 
    Given a query rectangle $Q$, we say that it \emph{hits} a grid cell $S$ if a corner of $Q$ is contained in $S$. 
    Note that the algorithm cannot get any information for the segments not hit by queries. 
    In particular, for a cell $S$ and a rectangle $Q$,
    the number of red (and blue) points in $S\cap Q$ is the same
    unless a corner of $Q$ is contained in $S$. 
    Then the remaining part of this proof is exactly the same as before. 
    We consider the event that an instance using the far gadget is chosen. This even happens with probability $1/2$ by construction. Thus to get a success probability at least $2/3$, we need to detect the far gadget. 
    Each query hits at most four grid cells. Therefore, 
    the probability that one range counting query hits the far gadget is $4/(16s)=1/(4s)$.     
    The probability that 
    at least one of the $s$ queries hits the far gadget is $1-(1-(4s)^{-1})^{s} < 1/3$.
    Since the probability that an instance using the far gadget is chosen is $1/2$, the total failure probability is larger than $2/3$, and thus the lemma holds.
\end{proof}

\subparagraph{Extension to a higher dimensional space.}
The construction of an input distribution used in the two-dimensional case can be easily extended to a higher dimensional space. We again define two gadgets on a $d$-dimensional cube $S$ of side length $\Delta/(4s^{1/d})$.
For the far gadget, we put two copies of the far gadget we constructed from the $(d-1)$-dimensional case, one on one facet of $S$ and one on its parallel facet of $S$. 
We do the same for the near gadget using two copies of the near gadget we constructed from the $(d-1)$-dimensional case. 

Now we define a distribution $\mu$ of instances as follows.
We construct the grid on $[\Delta]^d$ consisting of $16s$ cells  each of side length $\Delta/(4s^{1/d})$.  
We choose a cell $W$ as follows. 
With probability $1/2$, we let $W=\emptyset$. 
Then with probability $1/2$, we 
choose one cell from the $16s$ cells uniformly at random, 
and let it $W$. 
Then for each cell in the grid, except for $W$, we 
place the near gadget in the middle of the cell.
Then we place the far gadget in the middle of $W$. 
Notice that we do not place the far gadget anywhere if $W=\emptyset$. 
Thus it suffices to show that,
for any deterministic algorithm using $s$ range counting queries,
the probability that it
estimates the cost of $\opt$ within an additive error of $O(n\Delta/s^{1+1/d})$ on an instance chosen from $\mu$ is less than $2/3$. 
We can prove this similarly to Lemma~\ref{lem:distribution-2D}.
Here, the cost difference between two types of instances 
is $n/s$ times the side length of each cell, which is $O(n\Delta/s^{1+1/d})$. 
This is because, in a $d$-dimensional space, the side length of each cell is 
$O(\Delta/s^{1/d})$, and thus the lower bound on the additive error increases accordingly.

\section{Sampling a non-empty cell uniformly at random}
In this section, we show how to sample a grid cell (almost) uniformly at random from all grid cells containing points of $P$. Let $\mathcal G$ be a grid 
of certain side length $r$ imposed on the discrete space $[\Delta]^2$. 
Let $\mathcal C= \{\cell_1,\cell_2,\ldots,\cell_{m}\}$ denote the set of all grid cells of $\mathcal G$ containing points of $P$.
We say such cells are \emph{non-empty}.
For each cell $\cell_i$, we let $n_i$ denote the number of points contained in $\cell_i$. 
Here, we say sampling is $c$-\emph{approximate} uniform if
the probability that each non-empty cell is chosen is in
$[1/(cm), c/ m]$ 
for a constant $c$ with $c>1$.
Or, we simply say sampling is \emph{almost} uniform.

One might want to use the telescoping sampling of Lemma~\ref{lem:telescoping} directly: 
In particular, we choose a point $v$ of $P$ uniformly at random, and return the cell of $\mathcal G$ containing $v$.
However, the probability that a cell $c_i$ is chosen in this way is $n_i/n$, not $1/m$. 
%
%
So, instead, we use a two-step sampling procedure as stated in Algorithm~\ref{alg:cell:sampling}. 
Here, we are required to check if it has at most $\sqrt n$ non-empty cells. 
If so, we need to find all of them explicitly. 
This can be done $O(\sqrt n\log\Delta)$ range counting queries: 

\begin{restatable}{lemma}{heavycell}\label{lem:compute-all-nonempty}
 If the number of non-empty cells is $O(\sqrt n)$
 we can compute all of them using $O(\sqrt n\log\Delta)$ range counting queries.
\end{restatable}
\begin{proof}
    We may consider each grid cell as a single point. Then imagine that we construct a quadtree on the set of the points corresponding to the  grid cells. 
    Clearly, the height of the quadtree is $O(\log\Delta)$. 
    We traverse the quadtree along the paths towards the leaves corresponding to the non-empty cells. The total number of nodes of the quadtree we visited is $O(\sqrt n\log \Delta)$. Therefore, the lemma holds.
\end{proof}

\begin{algorithm}
  \caption{\sc{CellSampling}$(r)$ }
  \label{alg:cell:sampling}
  \begin{algorithmic}[1]
        \State Let $\mathcal{G}$ be the grid of side length $r$ imposed on the discrete space $[\Delta]^2$.
        \If {the number of non-empty cells of $\mathcal G$ is at most $\sqrt n$}
            \State Let $\mathcal C$ be the set of non-empty cells of $\mathcal G$.
            \State Return a cell chosen from $\mathcal C$ uniformly at random.
        \Else
            \State Select a set $T$ of $x= \sqrt{n}\log n$  \underline{points} from $P$ uniformly at random. 
                \State Let $n_p$ be the number of points of $P$ contained in the cell of $\mathcal{G}$ containing $p$.
                \State Assign weight $w(p) = \frac{n}{xn_p}$ to $p$.
        \State Let $m' = \sum_{p \in T} w(p)$ denote the total weight of all points in $T$.
        \State Sample a point $p$ from $T$
            with probability $\frac{w(p)}{m'}$.
    \State Return the cell of $\mathcal G$ containing $p$.
            \EndIf
  \end{algorithmic}
\end{algorithm}

\begin{restatable}{lemma}{cellsampling}\label{lem:cell sampling}
    {\sc{CellSampling}$(r)$}
    selects a cell almost uniformly at random from all non-empty cells of $\mathcal G$ using $\tilde O(\sqrt n)$ range counting queries.
\end{restatable}
\begin{proof}
    If $\mathcal G$ has $O(\sqrt n)$ non-empty cells, we can compute all of them deterministically using $\tilde O(\sqrt n)$ range queries by Lemma~\ref{lem:compute-all-nonempty}. Then we choose one of them uniformly at random. Thus it suffices to consider the case that $\mathcal G$ has $\Omega(\sqrt n)$ non-empty cells. 

     Observe that $m'$ is a random variable that corresponds to 
     the total weight of all points in the sampled set $T$. 
     Indeed, $m'$ estimates the number $m$ of non-empty cells in the grid $\mathcal G$. 
     We first show that the expected value of $m'$ is $m$ and with a high probability, the random variable $m'$ is sharply concentrated 
     around its expectation $m$.
    Suppose the sampled points are denoted as $p_1, \ldots, p_x$. For each sampled point $p_i$, we define a random variable $X_i = \frac{1}{n_{p_i}}$, where $n_{p_i}$ represents the number of points in the cell that contains $p_i$.
Recall that $\mathcal C= \{\cell_1,\cell_2,\ldots,\cell_{m}\}$ denotes the set of non-empty cells of $\mathcal G$.
Then, the expected value of $X_i$ is $\Ex{X_i} = \sum_{c \in C} \frac{n_c}{n}\cdot \frac{1}{n_c} = \frac{m}{n}.$ 
Now, the expected value of the random variable $X = \sum_{i=1}^x X_i$ is 
$\Ex{X} = \frac{xm}{n}$. 
Since $m' = X \cdot \frac{n}{x}$, we then have $\Ex{m'} = m$.

Now we use Chernoff bounds stated in Lemma~\ref{lem:cher}. Since $0\leq X_i\leq 1$, the probability that $|X-\Ex{X}|\geq  (\frac{mx}{n})/10$ is less than $\delta^{mx/n}\leq \delta^{\log n}$ for a constant $0<\delta<1$.
Recall that $x=\sqrt n\log n$ and $m\geq \sqrt n$ by the assumption. 
Therefore, we have $m'\in [9m/10, 11m/10]$ with probability at least $1-1/\delta$. 

In the following, assume that this event happens. 
For a fixed non-empty cell $c_i$ in the grid $\mathcal G$, we analyze the probability that $c_i$ is chosen.
Let $Y$ be the number of  sampled points in the first step contained in $c_i$.
Then $\Ex{Y}=\frac{xn_i}{n}$, where $n_i$ is the true number of points from $P$ in $c_i$.
By Chernoff bounds, 
the probability that $Y \notin [\frac{xn_i}{4n}, \frac{4xn_i}{n}]$ is again less than a small constant. 
Again, assume the event that $Y \in[\frac{xn_i}{4n}, \frac{4xn_i}{n}]$.
Then the total weight of all sample points contained in $c_i$ is in the range $[1/4,4]$. Thus the probability that $c_i$ is finally chosen is in the range $[1/4m',4/m']\subseteq [1/(cm), c/m]$ for a constant $c$. 
Therefore, {\sc{CellSampling}$(r)$}
returns a cell almost uniformly at random. 
\end{proof}


\begin{corollary}\label{cor:sample-estimator}
    We can estimate the number of non-empty cells within an $O(1)$-relative error using $\tilde O(\sqrt n)$ range counting queries. 
\end{corollary}

\subparagraph{Lower bound.} 
Now, we show that any randomized algorithm that can perform $c$-approximate uniform sampling for 
non-empty cells requires at least $\Omega(\sqrt{n}/c)$ range counting queries for any parameter $c\geq 1$. 
Our lower bound holds even in a discrete one-dimensional space $[\Delta]$ with $\Delta\geq n$. 
We assume that we have an interval $L$ of length $\Delta \ge n$ consisting of $n$ 
cells each of length $\Delta/n$. 
We subdivide $L$  into $\sqrt n/(4c)$ segments each of length $4c\Delta/\sqrt n$.  
Let $\mathcal S$ be the set of the segments for the line segment $L$. 
In this way, each segment in $S$ contains $(4c\Delta/\sqrt n)/(\Delta/n) = 4c\sqrt n$ cells. 
We construct a set $\mathcal I$ of $\sqrt{n}/(4c)+1$ instances as in Figure~\ref{fig:sampling-lowerbound}. 

\begin{figure}
    \centering
    \includegraphics[width=\linewidth]{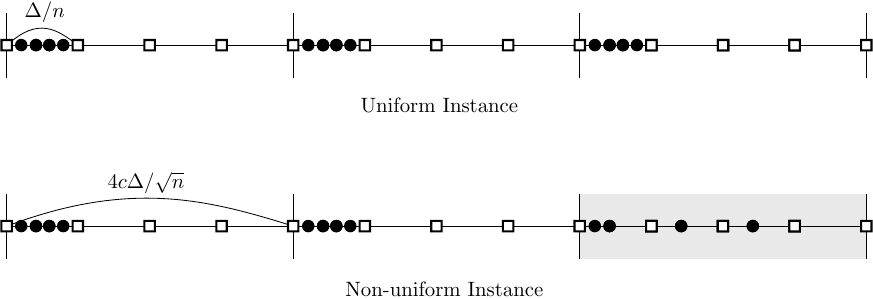}
    \caption{Illustration for the uniform instance and a non-uniform instance. The gray segment is the witness of the non-uniform instance.}
    \label{fig:sampling-lowerbound}
\end{figure}

We construct one \emph{uniform} instance of $L$. 
In this instance, for each segment $S \in \mathcal{S}$, we identify its leftmost cell and place $4c\sqrt{n}$ points in that cell. Therefore, 
each uniform instance contains $\sqrt{n}/(4c)$ non-empty cells (one corresponding to each segment) 
while maintaining a total of $n$ points per uniform instance.

Next, we construct $\sqrt{n}/(4c)$ \emph{non-uniform} instances of $L$. 
For the $i\textsuperscript{th}$ non-uniform instance, we select the $i\textsuperscript{th}$ segment, 
which we refer to as the \emph{witness segment}. In the leftmost cell of this witness segment, we place $2c\sqrt{n}$ points. Additionally, in each of the next $2c\sqrt{n}$ cells following the leftmost cell, we place one point per cell. 
For each of the remaining $\sqrt{n}/(4c) - 1$ non-witness segments, we place $4c\sqrt{n}$ points in the leftmost cell of each segment, 
as we did for the uniform instance. 
In this way, every non-uniform instance has $2c\sqrt{n} + 1 + \sqrt{n}/(4c) - 1 = 2c\sqrt{n} + \sqrt{n}/(4c)$ non-empty cells and 
contains $4c\sqrt{n} \cdot (\sqrt{n}/(4c) - 1) + 2c\sqrt{n} + 2c\sqrt{n} = n $ points in total. 
Thus, the ratio of the number of non-empty cells between a non-uniform instance and a uniform instance is $(2c\sqrt{n} + \sqrt{n}/(4c))/(\sqrt{n}/(4c)) = 8c^2 + 1$. For both instances see Figure~\ref{fig:sampling-lowerbound}. 

\begin{restatable}{lemma}{samplingdistinguish}
\label{lem:sampling:distinguish}
Given an instance $I$ from $\mathcal{I}$, any randomized algorithm that determines whether $I$ is uniform with a success probability of at least $2/3$ requires $\Omega(\sqrt{n}/c)$ range counting queries.
\end{restatable}
\begin{proof}
    By Yao's Minimax Theorem, it suffices to construct a distribution $\mu$ over 
    the instances described above such that no deterministic algorithm can determine, 
    with probability at least $2/3$, whether an input drawn from $\mu$ is uniform while using fewer than $\sqrt{n}/(16c)$ range counting queries.
    We construct the distribution $\mu$ as follows:
    the probability that the uniform instance is chosen is $1/2$, and the probability that each non-uniform instance is chosen is $2c/\sqrt n$. 
    Since we have $\sqrt{n}/(4c)$ non-uniform instances, the probability that we pick a non-uniform instance is $1/2$. 

    Now, suppose there exists a deterministic algorithm $\mathcal{A}$ that, for an instance sampled from $\mu$, 
    can correctly determine with probability at least $\delta$ whether the instance is uniform or 
    non-uniform using fewer than $\sqrt n/(16c)$ 
    range counting queries. 
    Since the success probability must be at least $2/3$, 
    the algorithm $\mathcal{A}$ is required to provide the correct answer for the uniform instance.
    Otherwise, the success probability is less than $2/3$ what proves this lemma.
    Moreover, since $\mathcal{A}$ is deterministic, for any non-uniform instance sampled from 
    the distribution $\mu$, if $\mathcal{A}$ fails to identify the witness segment, 
    it incorrectly concludes that the instance is uniform, despite it being non-uniform.

    Condition on the event that a non-uniform instance is chosen from $\mu$. 
    We analyze the probability that $\mathcal{A}$ fails to identify the witness segment. 
    Note that $\mathcal{A}$ can detect the witness segment only if one of its queries hits the witness segment of this instance.  
    Here, we say a query \emph{hits} a segment if an endpoint of a query interval is contained in the segment.
    The probability that a single query made by $\mathcal A$ fails to identify the witness segment in 
    a non-uniform instance is $1 - \frac{4c}{\sqrt{n}}$.  
    If $\mathcal A$ performs at most $\frac{\sqrt{n}}{16c}$ queries, 
    then the probability that none of these queries identifies the witness segment is  
    $\left(1 - \frac{4c}{\sqrt{n}}\right)^{\sqrt{n}/(16c)} \ge e^{-1/4} \ge 0.77$.
    Thus, the probability that $\mathcal A$ detects the witness segment is at most $1 - e^{-1/4} \le \frac{1}{3}$,  
    which is less than the required success probability of $\frac{2}{3}$, leading to a contradiction.
\end{proof}

\begin{restatable}{lemma}{cellsamplinglb}\label{lem:cell sampling lowerbound}
    Any algorithm that can perform $c$-approximate uniform sampling for non-empty cells requires $\Omega(\sqrt n/c)$ range counting queries for any parameter $c\geq 1$. 
\end{restatable}
\begin{proof}
    Assume that we have a $c$-approximate uniform sampling algorithm $\mathcal A$.
    That is, the probability that each non-empty cell is chosen is in $[1/(cm), c/m]$, where $m$ is the number of non-empty cells. Consider the uniform instance $I$ and any non-uniform instance $I'$.
    Recall that the number of non-empty cells in $I$ is $\sqrt n/(4c)$, and the number of non-empty cells in $I'$ is at least $2c\sqrt n$. 
    Therefore, $\mathcal A$ must distinguish the two cases with probability at least $1/2$. 
    However, by Lemma~\ref{lem:sampling:distinguish},
    no such algorithm exists, which gives a contradiction. 
\end{proof}

\section{An algorithm for the minimum spanning tree problem}
\label{apd:mst:algo}
Given a point $P$ in $[\Delta]^2$, we present an algorithm for estimating the cost of a minimum spanning tree of $P$ 
    within a factor of $(1+\eps)$ with a constant probability using $\tilde O(\sqrt n)$
    range counting queries. 
Here, we assume that $\Delta=O(n/\eps)$. This assumption can be achieved in our case as follows. 
As preprocessing, we find the smallest
square containing $P$. We can do this using $O(\log \Delta)$ range counting queries. Then, the size of the square is at least $\opt$. We consider this cell as the new discrete space $[\Delta]^2$. In this way, we obtain $\opt\geq\Delta$.

Now consider the grid of $[\Delta]^2$ of side length $(\eps/n)\cdot \Delta$. Imagine that we move each point to the center of its grid cell. This increases the optimal cost of the minimum spanning tree by at most $(\eps/n)\cdot\Delta\cdot n\leq \eps\cdot \opt$, which is negligible. Then we can obtain a new discrete space of size $n/\eps\cdot n/\eps$ by scaling. In this way, we can achieve the assumption we made before. 
However, since we do not have direct access to $P$ in the range counting oracle model, we cannot explicitly move the points.  
We can do this implicitly because any query range $R$ used by our algorithm will be a cell of the quadtree. 
More specifically, we expand the range $R$ into a larger range $R'$, where the side length of $R'$ is scaled by a factor of $n/\eps$ relative to the side length of $R$.
Then we use $R'$ in the original domain to query the number of points in $R$. 
This is exactly the number we want to find if $R$ is a cell of the quadtree. 

\subsection{The quadtree-based spanner}
The spanner of $P$ we use in this section is based on the well-separated pair decomposition (WSPD) of $[\Delta]^2$ constructed using a quadtree~\cite{callahan1995decomposition}.
Let $\mathcal Q$ be a quadtree constructed on $[\Delta]^2$.
We call a pair $(c,c')$ of cells in $\mathcal Q$ a \emph{well-separated pair} if $\eps\cdot d(c,c')\geq \max\{r(c), r(c)\}$, where
$r(c)$ is the side length of a cell $c$ and $d(c,c')$ is the distance between the centers of $c$ and $c'$.
Then a \emph{well-separated pair decomposition} of a set $P$ of $n$ points
is a set
$\{(c_1,c_1'), \ldots,(c_s,c_s')\}$ of well-separated pairs such that
for any two points $x,y \in P$,
there is a unique pair $(c_i,c_i')$ with $x\in c_i$ and $y\in c_i'$.
Callahan and Kosaraju~\cite{callahan1995decomposition} showed that there is a well-separated pair decomposition of $P$ consisting of $O(\eps^{-2}n)$ pairs.

As the way in which we construct the WSPD in our sublinear algorithm for the  minimum spanning tree problem is important, 
we give its brief sketch.
Then we run $\textsf{WSPD}(c_r,c_r)$ with the root $c_r$ of $\mathcal Q$.
Then the resulting set is a well-separated pair decomposition of $P$. We denote it by $\mathcal W$.
As the number of pairs of $\mathcal W$ is $O(\eps^{-2} n)$,\footnote{To get this bound, we need to remove all redundant pairs from $\mathcal W$. In particular, two pairs $(c,c')$ and $(\bar c, \bar c')$ with $P\cap c=P\cap\bar c$ and $P\cap c'=P\cap \bar c'$ might be contained in $\mathcal W$ in our case.} we cannot compute all such pairs using a sublinear number of range counting queries. However, given a pair $(c,c')$ of cells of $\mathcal Q$, we can determine if it is in $\mathcal W$ using $O(1)$ range counting queries.

\begin{algorithm}[ht]
  \caption{\textsf{WSPD}$(c,c')$ }
  \label{alg:wspd}
  \begin{algorithmic}[1]
    \If{$c=c'$ and $r(c)=0$}
        {\Return $\emptyset$}
    \EndIf

    \If{$c$ or $c'$ is an empty set}
        {\Return $\emptyset$}
    \EndIf

    \If {$r(c') < r(c)$}
        {exchange $c'$ and $c$}
    \EndIf

    \If {$r(c') \leq \eps d(c,c')$}
        {\Return $\{\{c,c'\}\}$}
    \EndIf
    \State $c_1,c_2,c_3,c_4\gets$ the children of $c'$ in $\mathcal Q$
    \State
    \Return $\textsf{WSPD}(c_1,c) \cup\textsf{WSPD}(c_2,c)\cup \textsf{WSPD}(c_3,c)\cup \textsf{WSPD}(c_4,c)$

  \end{algorithmic}
\end{algorithm}

\begin{lemma}[\cite{har2011geometric}]\label{lem:degree}
    Each point of $[\Delta]^2$ is contained in $O(\eps^{-2}\log \Delta)$ pairs of $\mathcal W$.
\end{lemma}

\begin{lemma}[\cite{har2011geometric}]\label{lem:similar-size}
    For a pair $(c,c')\in\mathcal W$,
    we have $r(c')/2 \leq r(c) \leq 2r(c')$.
\end{lemma}


Now we are ready to define a quadtree-based spanner $S=(P,E(S))$ of a point set $P\subseteq[\Delta]^2$.
For each cell $c$ of the quadtree $\mathcal Q$, its \emph{representative}
is defined as the lexicographically smallest point in $c\cap P$.
For each pair $(c,c')$ of $\mathcal W$,
we add the edges between the representatives of $c$ and $c'$ to the edge set of the spanner $S$.
Here, the edge length is set to $d(c,c')$.

\begin{lemma}[\cite{har2011geometric}]
    The cost of the minimum spanning tree of $S$
    is at most $(1+\eps)$ times the cost of the Euclidean minimum spanning tree of $P$.
\end{lemma}

To estimate the cost of the minimum spanning tree of $P$, we traverse a subgraph of $S$ using range counting queries only. Here, we use the following lemma.

\begin{lemma}\label{lem:cell traversal}
    Consider a grid cell $c$ of $\mathcal Q$ of side length $r/2$.
    We can find all grid cells $c'$ of $\mathcal Q$ of side length $r/2$ such that there is an edge in $S$ between a point in $c$ and a point in $c'$ of length at most $r$ using $O(\eps^{-2})$ range counting queries.
\end{lemma}
\begin{proof}
    If $c$ and $c'$ are connected by an edge of length at most $r$,
    $d(c,c')\leq 3r$.
    Moreover, we have $O(1)$ grid cells $c'$ with $r(c)=r(c')$ and $d(c,c')\leq 3r$.
    For each such pair $(c,c')$, we check if
    they are connected by an edge of length at most $r$. If so, we say $(c,c')$ is a \emph{yes-instance}.
    Given two such cells $c$ and $c'$, we do the following.
    As $(c,c')$ is not well-separated for any $\eps\leq 3$, it is sufficient to check if there is a pair in $\mathcal W$
    consisting of cells, one contained in $c$ and one contained in $c'$, of distance at most $r$. This can be done by applying $\textsf{WSPD}(c,c')$, but the number of queries we use here can be as large as $\Theta(n)$. To handle this issue,
    we stop traversing $\mathcal Q$ when we reach a grid cell of size at most a certain threshold. For this purpose, we let $i^*$ be the index where $2^{i^*}\leq \eps r/10 < 2^{i^*+1}$.
    We conclude that $(c,c')$ is a yes-instance
    if there exists a pair $(\bar c,\bar c')$ of non-empty grid cells of side length at least $2^{i^*}$ with $\bar c\subseteq c$ and $\bar c'\subseteq c'$ such that either (1) $(\bar c, \bar c')\in\mathcal W$ and $d(\bar c,\bar c')\leq r$, or (2) $(\bar c,\bar c')$ is not well-separated, and $d(\bar c,\bar c')+r(\bar c)+r(\bar c')\leq r$.
    In this case, we say that such a pair $(\bar c, \bar c')$ is a \emph{witness}.
    Since there are $O(\eps^{-2})$ grid cells $\bar c$ of side length at least $2^{i^*}$ with $\bar c \subseteq c$ (or $\bar c \subseteq c'$),
    the number of candidate pairs for $(\bar c,\bar c')$ is $O(\eps^{-4})$.
    But by slightly modifying $\textsf{WSPD}(c,c')$, it is not difficult to see that we can complete this task using
    $O(\eps^{-2})$ range counting queries.


    Now we prove the correctness of the algorithm.
    First, consider the case that our algorithm finds a witness $(\bar c, \bar c')$.
    If $(\bar c, \bar c')\in\mathcal W$ and $d(\bar c,\bar c')\leq r$, then
    the edge introduced by this pair must
    connect $c$ and $c'$, and thus we are done.
    If $(\bar c, \bar c')$ is not well-separated and $d(\bar c, \bar c')+r(\bar c)+r(\bar c') \leq r$, then
    there must be a pair in $\mathcal W$ consisting of cells, one in $\bar c$ and one in $\bar c'$. Moreover,
    such cells have distance at most $r$, and thus $(c,c')$ is a yes-instance.

    Now consider the other direction.
    Assume that $(c,c')$ is a yes-instance.
    That is, assume that there is an edge of $S$ between a point in $a\in c$ and a point $a'\in c'$ of length at most $r$.
    Then consider the pair $(c_a,c_{a'})$ of $\mathcal W$ defining the edge $aa'$.
    As we observed at the beginning of this proof, we have $c_a \subseteq c$ and $c_{a'}\subseteq c'$.
    If $r(c_a)\geq \eps r/10$ and $r(c_{a'})\geq \eps r/10$, then we are done since the algorithm would have considered this pair as a candidate for a witness. Thus without loss of generality, we assume that
    $r(c_a)<\eps r/10$.
    Then $r(c_{a'})\leq \eps r/5$ by Lemma~\ref{lem:similar-size}.
    By the construction, neither $(\parent(c_a), c_{a'})$ nor $(c_a,\parent(c_{a'}))$ is contained in $\mathcal W$,
    where $\parent (c_a)$ denotes the parent cell of $c_a$ in the quadtree.
    Thus the distance between $\parent (c_{a'})$ and $c_a$ must
    be less than $\parent(c_{a'})/ \eps \leq r/5$.
    Therefore, $a$ and $a'$ is sufficiently close: their distance is at most $r/5+2\eps r/5 \leq 2r/5$.
    Now consider two cells $(\bar c, \bar c')$ of side length $2^{i^*}$ containing $a$ and $a'$.
    Then we have $d(\bar c, \bar c')+r(\bar c)+r(\bar c')\leq r$, and thus
    our algorithm will detect $(\bar c, \bar c')$ as a witness, and then conclude that $(c,c')$ is a yes-instance.
\end{proof}

\subsection{Reduction to the component counting problem}
\label{sub:component_counting}

We can reduce the minimum spanning problem on $S$ to the problem of counting connected components of a subgraph of $S$.
Let $\opt$ denote the cost of the minimum spanning tree of $S$.
Let $S_i$ denote the subgraph of $S$ induced by the edges of $S$ of length at most $(1+\eps)^{i}$ for $i=0,1,2,\ldots,\log_{1+\eps}(2\Delta)$.
Here, the vertex set of $S_i$ is $V$, and $S_i$ might have several isolated vertices. Let $c_i$ denote the number of connected components of $S_i$.
Due to the following lemma given by~\cite{czumaj09estim_weigh_metric_minim_spann}, it is sufficient to estimate the number of
connected components of $S_i$.
To make our paper complete, we give its proof here.

    \begin{lemma}[\cite{czumaj09estim_weigh_metric_minim_spann}]\label{lem:reduction}
    $\opt \leq n - \Delta + \eps \sum_{i=0}^{\log_{(1+\eps)}
    (2\Delta)-1} (1+\eps)^ic_i \leq (1+\eps) \opt$
\end{lemma}

\begin{proof}
    Let $M^*$ be the set of edges of a minimum spanning tree of $S$.
    Then we can partition the edges into subsets $M_0, M_1,\ldots, M_s$ with respect to the edge lengths where
    $M_i$ denotes the set of edges of $M^*$ of length in $[(1+\eps)^i, (1+\eps)^{i+1}]$.
    Recall that the minimum distance of $P$ is at least one, and the maximum distance of $P$ is at most $2\Delta$.
    Thus $s\leq \log_{(1+\eps)} (2\Delta)$.

    Now it suffices to analyze the size of each subset $M_i$.
    We claim that $|M_i|=c_{i-1}-c_{i}$.
    Consider the subgraph $S_i^*$ of $S_i$ induced by $M_0\cup M_1\cup\ldots\cup M_i$.
    Clearly, each component of $S_i^*$ must be a tree.
    Each component of $S_i^*$ is subdivided into subtrees in $S_{i-1}^*$. Here, the number of edges of $S_i^*$ not appearing in $S_{i-1}^*$ is exactly $c_{i-1}-c_{i}$ due to the cut property.
    Then we have the following.
    \[
    \opt \leq \sum_{i=0}^{s}(1+\eps)^i (c_{i-1}-c_{i}) \leq (1+\eps)\opt.
    \]
    Then we can reformulate the inequality as follows.
    \begin{align*}
        \sum_{i=0}^{s}(1+\eps)^i (c_{i-1}-c_{i})
        &=
        \sum_{i=0}^{s}(1+\eps)^i c_{i-1} - \sum_{i=0}^{s}(1+\eps)^i c_{i} \\
        &= n - \Delta + \sum_{i=0}^{s-1}((1+\eps)^i-(1+\eps)^{i-1}) c_i \\
        &= n - \Delta + \sum_{i=0}^{s-1} \eps(1+\eps)^i c_i
    \end{align*}
    Therefore, the lemma holds.
\end{proof}

In the following, we let $w=\log_{1+\eps}(2\Delta)$.
For each index $i$ with $i\leq w$, we estimate the number $c_i$ of components of $S_i$.
Chazelle et al.~\cite{chazelle05approx_minim_spann_tree_weigh_sublin_time} presented a sublinear-time algorithm for estimating the number of components of a graph with bounded average degree.
However, there is an issue: The additive error of CRT algorithm is $n/t$, where $t$ is the number of queries.
This additive error becomes $(1+\eps)^i\cdot (n/t)$ for the estimator of the cost of a minimum spanning tree by Lemma~\ref{lem:reduction}. If $i$ is large, this error cannot be negligible.
These issues can be easily resolved by applying CRT algorithm
on a graph $S_i'$ obtained by contracting several vertices of $S_i$. This approach is also used in~\cite{frahling2005sampling} to design dynamic streaming algorithms for the minimum spanning tree problem.
The construction of $S_i'$ uses the following lemma.

\begin{lemma}\label{lem:non-empty}
    Let $c$ be a cell of the quadtree with $r(c)\leq (1+\eps)^i/4\leq 2r(c)$. Any two vertices in $c\cap V(S_i)$ are connected in $S_i$.
\end{lemma}
\begin{proof}
    Assume to the contrary that $u$ and $v$ are contained in different components, say $C_u$ and $C_v$, in $S_i$.
    Then there must be a path in $H$ between $C_u$ and $C_v$ consisting of edges of length larger than $(1+\eps)^i$.
    Note that the edge weight of $uv$ is at most $2r(c)\leq (1+\eps)^i/2$, and this violates the cut property.
    Therefore, the lemma holds.
\end{proof}

Then we construct $S_i'$ as follows.
Let $\mathcal G$ be the grid of side length $(1+\eps)^{i^*}/4$.
For each cell $c$ of $\mathcal G$, we contract all vertices in $c\cap V(S_i)$ into a single vertex. After we do this for all cells of $\mathcal G$, then we remove all parallel edges and loops from the resulting graph. Notice that
as the purpose of $S_i'$ is to estimate
the number of components of $S_i'$, we do not need to care about edge weights of $S_i'$.
The number of connected components remains the same.
In the following, we estimate the number of components of $S_i'$ for each index $i\in \{0,1,2,\ldots,w\}$.

\subsection{Estimating the number of components}
Now we are ready to estimate the number $c_i$ of connected components of $S_i$ for an index $i$.
We use the algorithm by Chazelle et al.~\cite{chazelle05approx_minim_spann_tree_weigh_sublin_time} by setting parameters properly.
This algorithm estimates the number of connected components of a graph with a bounded average degree.
In our case, the maximum degree of $S_i'$ is also bounded.
Thus we can simplify CRT algorithm as we will describe below.
The algorithm is based on the following simple observation.
Let $G=(V,E)$ be a graph.
For a vertex $v\in V$, let $d(v)$ denote the degree of $v$ in $H$. Let $m(v)$ denote the number of edges in the connected component of $G$ containing $v$. Then let $\beta(v)=d(v)/(2m(v))$ if $v$ is not isolated. If $v$ is isolated in $G$, then we let $\beta(v)=1$.
\begin{lemma}[{\cite{chazelle05approx_minim_spann_tree_weigh_sublin_time}}]
    The number of connected components of $G$ is $\sum_{v\in V} \beta(v)$.
\end{lemma}
CRT algorithm takes $r=O(1/\eps^2)$ samples $v_1,v_2,\ldots,v_r$ from $V$. Then for each sampled vertex $v_j$, it estimates $\beta_j=\beta(v_j)$. Let $\hat\beta_j$
be the estimator of $\beta_j$. Then this algorithm simply returns $(n/r)\cdot \sum_{i=1}^r \hat{\beta_i}$.
However, if the connected component of $G$ containing $v_i$ contains a large number of vertices, estimating $m(v_i)$ requires
a large number of queries.
Thus it only counts the number of connected components of $G$ of size less than a certain threshold. This induces the additive error in the estimator for the number of connected components of $G$. If the maximum degree of $G$ is constant,
one can compute $\beta_j$ exactly and deterministically in  the case that its connected component in $G$ has size $\sqrt n$.

\subparagraph{Algorithm.}
First, we set the threshold to $t=\sqrt n$.
We pick $r=O(1/\eps^2)$ random samples
$\cell_1,\cell_2,\ldots,\cell_r$ from $V(S_i')$ uniformly at random.
We call them the \emph{seeds}.
Recall that a vertex of $V(S_i')$ corresponds to a non-empty grid cell of $\mathcal G$, where $\mathcal G$ is the grid of side length $(1+\eps)^{i}$. This can be done using $O(\sqrt n\log n)$ range counting queries as in Algorithm~\ref{alg:cell:sampling}. 
This sampling algorithm also gives an estimator $\hat n$ of the number of non-empty cells of $\mathcal G$ within $O(1)$ relative error by Corollary~\ref{cor:sample-estimator}.

For each seed $\cell_j$, let $v_j$ denote the vertex of $S_j'$ corresponding to $\cell_j$. We first compute the
degree $d_j'$ of $v_j$ using $O(\eps^{-2})$ range counting queries using Lemma~\ref{lem:cell traversal}. 
To compute $m_j'$, we
apply BFS on $S_i'$ from $v_j$ until we visit $t$ vertices of $S_i'$. 
Note that the edge set of $S_i$ is exactly the set of all edges of $S$ of length $[1,(1+\eps)^i]$, and the side length of a cell of $\mathcal G$ is at most $(1+\eps)^i/2$.
Thus we can use Lemma~\ref{lem:cell traversal} to apply BFS on $S_i'$.
There are three possible cases.
If traversal visits more than $t$ vertices of $S_i'$, we let $\beta_j'=0$. This will induce a small additive error as we will see later.
If $v_j$ is isolated in $S_i'$, let $\beta_j'=2$.
If BFS successfully terminates, we enumerate its all neighbors in $S_i'$ of the vertices visited by the BFS. In this way, we can compute the number $m_j'$ of edges in the component of $S_i'$, and thus we can obtain $\beta_j'$.
Then we return $\hat c_i= (\hat n/r)\cdot \sum_{i=1}^r \beta_i'$.

\begin{lemma}\label{lem:component-analysis}
    We have $|(1-\eps)c_i-\hat{c_i}| \leq  \opt/(1+\eps)^i$ with probability at least $2/3$. 
\end{lemma}
\begin{proof}   
    In the following, we assume the event that $\hat n\in [n^*/2,2n^*]$, where $n^*$ is the number of vertices of $S_i'$, which occurs with a constant probability. 
    We first analyze the expected value $\Ex{\hat c_i}$ of $\hat c_i$.
    We say a component of $S_i'$ is \emph{small} if it has at most $t$ vertices (cells), and \emph{large} otherwise. 
    Note that the number of large components is $O(n^*/t)$.
    Let $Y_j$ denote the random variable where, for the $j$th random seed $v_j$,
    $Y_j= d'(v)/(2m'(v))$ if $v_j$ is contained in a small component,
    $Y_i=1$ if it is isolated,
    and $Y_i=0$ otherwise.
    Recall that seeds are chosen from $V(S_i')$ uniformly at random.
    Let $C_\text{large}$ denote the set of vertices of $S_i'$
    contained in the large components of $S_i'$.
    First, we can easily see that
    $\Ex{\hat{c_i}}\leq c_i$.
    Then we analyze its lower bound.
\[
   \Ex{\hat{c_i}} \geq \frac{\hat n}{r} \sum_j \Ex{Y_j} = \frac{\hat n}{\hat n-|C_\text{large}|} \sum_{v\notin C_\text{large}} \frac{d'(v)}{2m'(v)}
    \geq  c_i - \hat n/t.
\]
    Here, notice that $n^*\leq 4\cdot \opt/(1+\eps)^i$ if $n^*\geq 4$. To see this,
    consider the adjacency graph between the non-empty cells of the grid $\mathcal G$: a vertex corresponds to a grid cell, and two vertices are adjacent if they share a common edge on their boundaries. Then this graph is planar, and thus it has a 4-coloring.
    Thus it has an independent set of size $n^*/4$.
    The distance between any two points from different cells in the independent set is at least $(1+\eps)^i$, and thus $\opt\geq (n^*/4)\cdot (1+\eps)^i$.
    Therefore, the additive error in the lower bound is at most $\opt/(1+\eps)^i$. Finally,
    we have
    \[
   \Ex{\hat{c_i}}  \geq  c_i - \hat n/t_i \geq c_i-O(1)\cdot \opt/(1+\eps)^i.
\]

    Now we can analyze the probability that $|\hat c_i- c_i| \geq \opt/(1+\eps)^i$ using Chebyshev's inequality in Lemma~\ref{lem:chebyshev}.
    Here, for each random variable $Y_j$, we have $0\leq Y_i \leq 1$. Let $Y=Y_1+\ldots+Y_r$.
    Therefore,
    \[
    \Var{[Y_i]} \leq \Ex{Y_i^2} \leq \Ex{Y_i} \leq c_i/n^*, \text{ and }
    \Var{[Y]} = \sum_i \Var{[Y_i]} \leq (rc_i)/n^*.
    \]
    Therefore, since $r=O(\eps^{-2})$ we have
    \[
    \Var{[\hat c_i]} = \Var{[(\hat n/r) \cdot Y]} \leq (\hat n^2 /r^2 )\cdot (rc_i)/{n^*} \leq 4c_in^*\eps
    \]    
    Recall that $|\Ex{\hat c_i} - c_i| \leq O(1)\cdot \opt/(1+\eps)^i$. 
    By Chebyshev's inequality, for any $a\geq 1$, 
    the probability that $|\hat c_i - \Ex{\hat c_i}|$ at least $O(\opt)/(1+\eps)^i$ is at most
    \[
     (4c_in^*\eps) \cdot (1+\eps)^{2i}/O(\opt^2) \leq \eps (c_i\cdot (1+\eps)^i /O(\opt)) \leq O(\eps).
\]
Therefore, we have the claimed bound.
\end{proof}


\subsection{Putting it all together}

In this subsection, we summarize our argument and analyze the approximation factor of the estimator for the Euclidean minimum spanning tree.
For each index $i\in \{0,1,\ldots,\log_{1+\eps}(2\Delta)\}$,
let $S_i$ denote the subgraph of $G$ induced by the edges of $E(S)$ of length in $[1,(1+\eps)^{i+1}]$.
Then we estimate the number of components of $S_i$ (and $S_i'$) using $\tilde O(\sqrt n)$ range counting queries. Let $\hat c_i$ denote the estimator of $c_i$.
Then we return $\hat L= n - \Delta + \eps \sum_{i=0}^{\log_{(1+\eps)}
    (2\Delta)} (1+\eps)^i\hat c_i$.

\begin{lemma}
\label{lem:emst_proof}
    With a constant probability, we have 
    $(1-O(\eps))\cdot \opt \leq \hat L \leq  (1+O(\eps)) \cdot \opt$,
    where $\opt$ denotes the cost of a minimum spanning tree of $P$. 
\end{lemma}
\begin{proof}
    By Lemma~\ref{lem:reduction}, it suffices to show that
    \[
    \Ex{\sum_{i=0}^{\log_{1+\eps}\Delta} (1+\eps)^i \hat c_i} = (1+\eps)\sum_{i=0}^{\log_{1+\eps}(2\Delta)} (1+\eps)^ic_i + \eps \cdot n.
    \]
    This is because $\opt \geq n$. The additive error of $\eps\cdot n$ is at most $\eps\cdot \opt$, which is negligible.
    If we look at the proofs of the lemmas more carefully,
    by increasing the number of samples chosen by each step by a factor of $O(\log\Delta)$, we can increase the success probability to $1-1/\Delta$. Then with a constant probability, we have the desired bound for all indices $i$. 
    We assume such an event occurs. 
    For all indices $i$, we have
    $\hat c_i\in [c_i - n/(1+\eps)^i, c_i+n/(1+\eps)^i]$ 
    by the proof of Lemma~\ref{lem:component-analysis}.
    Since we have $\log_{1+\eps}\sqrt n$ such indices, the additive error that comes from all indices will be at most $\eps\cdot n$. So far, we have focused on the cost of the spanner. However, we need to state our result with respect to the cost of a minimum spanning tree of $P$. Since     
    the cost of a minimum spanning tree of the spanner is $(1+\eps)\cdot \opt$, this 
    induces a factor of $(1+\eps)$, 
    and thus the lemma holds.
\end{proof}

Therefore, we have the following theorem.
\begin{theorem}
    \label{thm:mst}
    Given a set $P$ of size $n$ in a discrete space $[\Delta]$, 
    we can estimate the cost of a minimum spanning tree of $P$ 
    within a factor of $(1+\eps)$ with a constant probability using $\tilde O(\sqrt n)$
    range counting queries. 
\end{theorem}

\section{Lower bound for the minimum spanning tree problem}
In this section, we show that 
any randomized constant-factor approximation algorithm for the Euclidean minimum spanning tree problem uses $\Omega(n^{1/3})$ range counting queries. 
For this, we construct a distribution $\mu$ of instances
where any randomized algorithm using $o(n^{1/3})$ queries fails to obtain a constant-factor approximate solution with  probability at least $1/3$. 

Let $[\Delta]^2$ be a discrete domain with $\Delta=O(n)$. 
We subdivide $[\Delta]^2$ into $16n^{1/3}$ equal-sized cells where each has side length $4n^{5/6}$. 
See Figure~\ref{fig:mst-lowerbound}. 
In the middle of each cell, we place either the \emph{strip} gadget or the \emph{uniform} gadget. 
Each gadget is defined on the domain $[n^{5/6}]^2$ which we further subdivide into $n^{4/3}$ finer cells
of side length $n^{1/6}$. 
For the strip gadget, we put one point to each finer cell on one diagonal of $[n^{5/6}]^2$. 
For the uniform gadget, we put one point to 
the cell in the $i$-th row and the $((i\cdot n^{1/3})\mod n^{2/3})$-th column
for $i=0,1,2,\ldots,n^{2/3}$.
In total, each gadget contains $n^{2/3}$ points. 
Notice that the cost of a minimum spanning tree of the points inside the strip gadget is $\Theta(n^{5/6})$, and 
the cost of a minimum spanning tree of the points inside the uniform gadget is $\Theta(n^{7/6})$. 
Now we define a distribution $\mu$ of instances.
First, with probability $1/2$, we place copies of the strip gadget in all cells of the domain $[\Delta]^2$. 
Then with probability $1/2$, we pick one cell $c$ of the domain $[\Delta]^2$ uniformly at random.
We place a copy of the uniform gadget in $c$, and place copies of the strip gadget in the other cells.
Then we have the following lemma. 
\begin{lemma}\label{lem:lb-cost-mst}
    Let $I$ and $I'$ be two instances of $\mu$ such that $I$ has the uniform gadget, but $I'$ does not use the uniform gadget. Then $\mst(I') \leq 2\cdot \mst(I)$,
    where $\mst(\cdot)$ denotes the cost of a minimum spanning tree.
\end{lemma}
\begin{proof}
    Let $M$ be a minimum spanning tree of $I$.
    Observe that the cost of $M$ is at least $2\cdot n^{7/6}$. 
    This is because $M$ contains at least $n^{1/3}$ edges connecting two points contained different cells in the domain $[\Delta]^2$. Their total length is at least $2\cdot n^{7/6}$. 

    Now we analyze an upper bound on the cost of a minimum spanning tree of $I'$. We can construct a spanning tree of $I'$ from $M$ as follows.
    The two instances are the same, except that $I$ has a uniform gadget in a cell, say $c$,
    and $I'$ has a strip gadget in $c$. 
    By the cut property, there are at most $O(1)$ edges of $M$ 
    having one endpoint in $c$ and one endpoint lying outside of $c$. Moreover, such edges have length at most $8n^{5/6}$. For each such edge, we reconnect it with any vertex in $c$. Then 
    we remove all edges having both endpoints in 
    $c$ from $M$,
    and add the edges of the minimum spanning tree of 
    the strip gadget. 
    In total, the cost of $M$ decreases by 
    at most $n^{7/6}- O(1)\cdot n^{5/6}$. Here, the term $n^{7/6}$ is the cost of the uniform gadget. 
    Therefore, we have
    $\displaystyle \mst(I')-\mst(I) \leq n^{7/6} - O(1)\cdot n^{5/6} \leq (1/2)\cdot n^{7/6} \leq \mst(I)$.
\end{proof}

\begin{figure}
    \centering
    \includegraphics{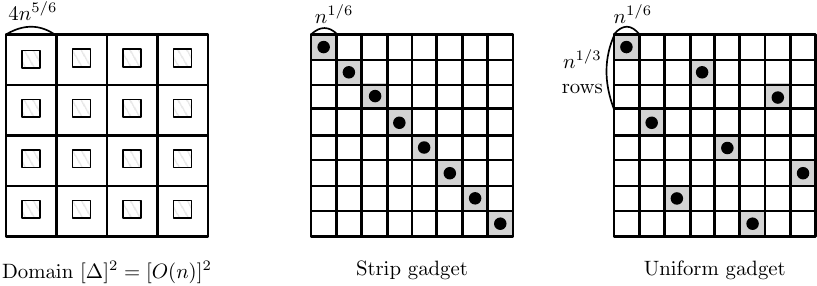}
    \caption{The domain $[\Delta]$ is partitioned into $16n^{1/3}$ cells.
    Each cell contains the strip gadget or the uniform gadget.
    The strip gadget has cost $\Theta(n^{5/6})$ while the uniform gadget has cost $\Theta(n^{7/6})$.}
    \label{fig:mst-lowerbound}
\end{figure}

\begin{restatable}{lemma}{mstlb}\label{lem:distribution-mst}
    Any randomized constant-factor approximation algorithm for the minimum spanning tree problem on a point set of size $n$ in a discrete space $[\Delta]^d$ requires $\Omega(n^{1/3})$ range counting queries.
\end{restatable}
\begin{proof}
    Due to Yao's Minmax theorem,
    it suffices to show that for any deterministic algorithm using $n^{1/3}$ range counting queries, 
    the probability that it estimates the cost of an instance chosen from $\mu$ 
    within a factor of $1.5$ is less than $2/3$. 
    
    Assume to the contrary that such a deterministic algorithm $\mathcal A$ exists. By Lemma~\ref{lem:lb-cost-mst}, it can distinguish the instance containing the uniform gadget and the instances not containing the uniform gadget.
    
    We can consider $\mathcal A$ as a decision tree
    of depth $n^{1/3}$. Here, each node corresponds to an orthogonal range, and each leaf corresponds to the output. 
    Given a query rectangle $Q$, we say that it \emph{hits} a cell $c$ of the domain $[\Delta]^2$ if a corner of $Q$ is contained in $c$. 
    Note that the algorithm cannot get any information for the cells not hit by queries. 
    In particular, for any cell $c$ and any query rectangle $Q$, the number of points in $c\cap Q$ remains the same unless $c$ contains a corner of $Q$. 
    On the other hand, $\mathcal A$ uses only $n^{1/3}$ range counting queries, and thus the probability that any of the $n^{1/3}$ range queries hits the uniform gadget is less than $2/3$.

    More formally, we consider the event that an instance using the uniform gadget is chosen from $\mu$. This even happens with probability at least $1/2$ by construction. Thus to get a success probability at least $2/3$, we need to detect the uniform gadget. 
    Recall that we can detect the uniform gadget only when it is hit by $Q$.  On the other hand, each query hits at most four cells. Notice that the probability that a fixed query range hits the uniform gadget is $1/(4n^{1/3})$ as the total number of cells in $[\Delta]^2$ is $16n^{1/3}$. 
    Therefore, the probability that 
    at least one of the $n^{1/3}$ queries hits the far gadget is $1-(1-1/(4n^{1/3}))^{n^{1/3}} < 1/3$. In other words, the failure probability is at least $2/3$ assuming that we are given an  instance using the uniform gadget. 
    Since the probability that an instance using the uniform gadget is chosen is $1/2$, the total failure probability is larger than $2/3$, and thus the lemma holds.
\end{proof}

\bibliography{references}
\appendix

\section{Concentration bounds}
\begin{lemma}[Chernoff Bound~\cite{Cher}]
\label{lem:cher}
Let $Y_1,\cdots,Y_m$ denote $m$ independent Poisson trials such that $\Pr{Y_i=1}=p_i$ for $1\leq i\leq m$. Let $Y=\sum_{i=1}^m Y_i$ and $\mu=\Ex{Y}$. For $0<\delta<1$, 
\[
  \Pr{|Y-\mu|  \geq \delta\mu}\leq 2e^{-\mu\delta^2/3}
\]
\end{lemma}

\begin{lemma}[Markov's Inequality]
\label{lem:markov}
Let $X$ be a random variable that assumes only nonnegative values. Then, for all $a>0$, 
\[
  \Pr{X\geq a }\leq \frac{\Ex{X}}{a}
\]
\end{lemma}



\begin{lemma}[Chebyshev's Inequality]
\label{lem:chebyshev} Let $X$ be a random variable. For any $a>0$, we have 
\[
  \Pr{|X-\Ex{X}|\geq a}\leq \frac{\Var{X}}{a^2}.
\]
\end{lemma}

\end{document}